\documentclass{article}
\usepackage[a4paper, margin=1in]{geometry}
\usepackage{color}
\usepackage{booktabs}
\usepackage{comment}
\usepackage{amsmath}
\usepackage{amssymb}
\usepackage{tabularx}
\usepackage{hyperref}
\usepackage{authblk}
\usepackage{caption}
\usepackage{subcaption}
\usepackage{graphicx}
\usepackage{algpseudocode}
\usepackage[ruled,vlined]{algorithm2e}
\usepackage[skip=0.7em]{parskip}
\usepackage{xcolor}
\usepackage{soul}
\usepackage{enumitem}
\usepackage{amsthm}
\newtheorem{theorem}{Theorem}
\newtheorem{proposition}{Proposition}

\newtheorem{lemma}{Lemma}
\newtheorem{claim}{Claim}
\newtheorem{corollary}{Corollary}

\hypersetup{
	unicode=false,          
	pdftoolbar=true,        
	pdfmenubar=true,        
	pdffitwindow=false,     
	pdfstartview={FitH},    
	pdftitle={My title},    
	pdfauthor={Author},     
	pdfsubject={Subject},   
	pdfcreator={Creator},   
	pdfproducer={Producer}, 
	pdfkeywords={keywords}, 
	pdfnewwindow=true,      
	colorlinks=true,        
	linkcolor=blue,         
	citecolor=blue,         
	filecolor=magenta,      
	urlcolor=cyan           
}

\title{
    \textbf{Social welfare optimisation under institutional reward and punishment}
}
\author[1,2]{Van An Nguyen}
\author[1,2]{Vuong Khang Huynh}

\author[1,2]{Huu Loi Bui}
\author[1,2]{Hai Anh Ha}
\author[1,2]{Quang Dung Le}
\author[1,2]{Tan Dat Nguyen}
\author[1,2]{Ngoc Ngu Nguyen}
\author[3]{Zhao Song}
\author[4]{Manh Hong Duong} 
\author[1,2,$\star$]{Le Hong Trang}
\author[3,$\star$]{The Anh Han }

\affil[1]{Faculty of Computer Science and Engineering, Ho Chi Minh City University of
Technology (HCMUT), Vietnam}
\affil[2]{Vietnam National University - Ho Chi Minh City (VNU-HCM), Vietnam}
\affil[3]{School of Computing, Engineering and Digital Technologies, Teesside University, United Kingdom}
\affil[4]{School of Mathematics, University of Birmingham, Birmingham, United Kingdom}

\affil[$\star$]{Corresponding authors: The Anh Han (Email: t.han@tees.ac.uk), Le Hong Trang (Email: lhtrang@hcmut.edu.vn)}
\date{ }

\begin{document}

\maketitle

\begin{abstract}

Institutional incentives are widely used to promote cooperation among autonomous, self-regarding agents, from human societies to multi-agent and AI systems. Existing work typically treats incentive design as a bi-objective problem: minimise institutional cost while achieving a high long-run frequency of cooperation. Whether such schemes also maximise social welfare—total population payoff net of institutional expenditure—has remained largely unexplored.
We develop a welfare-centric framework for institutional incentives in finite, well-mixed populations playing a social dilemma (Donation Game and Public Goods Game), considering both rewards for cooperators and punishments for defectors. For each mechanism, we derive explicit expressions for expected social welfare and characterise how it depends on incentive efficiency and selection intensity. Analytically, we identify parameter regimes where social welfare has a single optimal incentive level and regimes with qualitative phase transitions, in which welfare becomes non-monotonic with multiple local optima. We prove that any welfare-maximising incentive is either zero or  concentrated around a simple closed-form target, and we provide an efficient algorithm to compute these optima. Comparing reward and punishment, we further derive close-formed conditions under which  reward outperform punishment in terms of social welfare for any given budget.
Overall, our results reveal a systematic gap between incentives optimised for cost or cooperation frequency and those that maximise welfare. 

\end{abstract}
\tableofcontents

\section{Introduction}
The evolution of cooperation has long been a central puzzle in evolutionary biology, the social sciences, and multi-agent systems \cite{sigmund2010calculus,han2022emergent,paiva2018engineering,tuyls2007evolutionary,perc2017statistical}. In many strategic settings, such as one-shot interactions in the Prisoner’s Dilemma and the Public Goods Game \cite{archetti2012game,axelrod1981evolution,doebeli2005models}, classical evolutionary theory predicts that selection on individual fitness typically favours selfish behaviour \cite{sigmund2010calculus,hofbauer1998evolutionary}. Yet, cooperative behaviour remains widespread in both humans and other animals \cite{fehr2004social,nowak2006five}. This apparent tension between individual and collective benefits has motivated extensive research into mechanisms that enable, stabilise, and shape cooperation in social dilemmas \cite{perc2017statistical,HAN2026_social_welfare,szabo2007evolutionary,xia2023reputation,liu2026dynamic}.

To address this puzzle \cite{perc2017statistical,rand2011evolution}, numerous mechanisms have been proposed, including spatial structure \cite{szabo2007evolutionary,szolnoki2011phase}, kin and group selection \cite{hamilton1964genetical,traulsen2006evolution}, direct and indirect reciprocity \cite{nowak:2005:nature,xia2023reputation}, and institutional incentives \cite{fehr2000cooperation,sasaki2012take,chen2015first,sun2023state}. A particularly important line of work studies external institutional influence, where a central authority invests resources to steer populations towards cooperation \cite{chen2015first,duong2023cost,sigmund2001reward}. For example, an external institution, such as an international body or a local authority, can conditionally reward cooperative individuals or punish defective ones  based either on global statistics  or  local spatial neighbourhood information  \cite{cimpeanu2021cost,han2018fostering,wang2019exploring,sigmund2010social,brandt2006punishing,wang2021incentive,sun2023state}. 

In well-mixed populations, the analysis of these institutional interventions has traditionally been framed as a bi-objective optimisation problem \cite{han2018interference,duong2023cost}. In this framework, a decision-maker conditionally rewards cooperators to guarantee a desired cooperation level while simultaneously minimising the interference cost. Subsequent work has developed rigorous stochastic analyses of such institutional incentives, characterising optimal schemes across different selection intensities and revealing sharp phase transitions in cost efficiency \cite{duonghan2021cost,duong2026cost}. Consequently, prior models predominantly evaluate the success of external interventions through these two specific performance criteria: maximising the frequency of cooperation and minimising the institutional investment \cite{wang2019exploring,gross2025hidden}.

Despite these analytical advances, existing models evaluating external interventions largely overlook a crucial holistic metric: \textit{social welfare}. Social welfare captures the system-level impact of these interventions, broadly defined as the total population payoff minus the external institutional investment \cite{HanInterface2024Welfare}. From a societal perspective, cooperation is valuable only insofar as it results in net benefits once the costs of promoting it are taken into account \cite{karsu2015inequity,nyman2006efficiency,kaneko1979nash,song2026emergence}. Focusing solely on the dual metrics of cooperation levels or budgetary savings risks endorsing interventions that are formally successful by those measures yet ultimately detrimental to overall welfare \cite{HAN2026_social_welfare}. Addressing this gap is essential not only for theoretical completeness but also for practical relevance, especially in contexts such as public policy, distributed systems, and resource allocation \cite{paiva2018engineering,nyman2006efficiency,HAN2026_social_welfare}.

Therefore, in this paper, we \textit{analytically} study the optimisation of social welfare under costly institutional incentives in well-mixed populations. We examine whether schemes that minimise cost and maximise cooperation also maximise social welfare, and whether welfare‑maximisation leads to smaller, larger, or qualitatively different investments than cost‑focused approaches..

We explicitly embed social welfare into the analytical framework of institutional incentives in finite well-mixed populations. This enables a systematic comparison between cost-based and welfare-based optima and reveals when they coincide and when they diverge, including cases where maximising cooperation or minimising cost leads away from welfare-maximising policies. We adopt the well-established social dilemmas games, namely the Donation Game and the Public Goods Game \cite{sigmund2010calculus}, studying how social welfare optimisation responds to variations in institutional incentives over a broad parameter range. By treating social welfare as the central optimisation objective, this work provides a more realistic and policy-relevant framework for understanding how cooperative behaviour should be engineered in biological, social, and artificial systems \cite{HAN2026_social_welfare}.

\subsection*{Organisation of the paper} The rest of the paper is organized as follows. In Section \ref{sec: models and methods}, we introduce the evolutionary game models and mathematically formulate the problem of optimising the social welfare for both institutional reward and punishment. Section \ref{sec: main results} presents the main analytical results while Section \ref{sec: numerics} provides numerical findings and validations. Further discussion and outlook is given in Section \ref{sec: summary}. Detailed proofs of the main analytical  results are provided in Appendix, namely Appendix \ref{sec: appendix1} (for reward) and Appendix \ref{sec: appendix2} (for punishment). Additional numerical simulations with different parameters' values to show the robustness of our results are also provided in Appendix \ref{sec: appendix2}.

\section{Models and Methods}
\label{sec: models and methods}
This section describes the evolutionary game model used to formulate the bi-objective optimisation of institutional incentives \cite{duonghan2021cost,han2018interference}. We then derive the social welfare functions for both institutional reward and punishment for a well-mixed population setting.

\subsection{Evolutionary game models}

We consider a finite, well-mixed population of $N$ players, who interact with each other using a cooperation dilemma, namely the Donation Game or Public Goods Game. The population evolves according to the Fermi strategy update rule~\cite{sigmund2010social}: a player $X$ with fitness $f_X$ adopts the strategy of another player $Y$ with fitness $f_Y$ with probability
\[
P_{X,Y}=\left(1+e^{-\beta(f_Y-f_X)}\right)^{-1},
\]
where $\beta>0$ denotes the selection intensity.

The population dynamics are modelled as an absorbing Markov chain over the state space $\{S_0,S_1,\dots,S_N\}$, where $S_i$ represents the state with $i$ cooperators. The homogeneous states $S_0$ and $S_N$ are absorbing, while $S_1,\dots,S_{N-1}$ are transient. Let $\Pi_C(i) \text{ and } \Pi_D(i)$ be the expected payoffs of a cooperative player (C\text{-}player) and a defective player (D\text{-}player) in state $S_i$ of the population, respectively. Let $U=\{u_{ij}\}_{i,j=1}^{N-1}$ be the transition matrix among transient states. In the absence of institutional incentives, for $1\le i\le N-1$, the transition probabilities are
\begin{equation}
\label{eq: transition probabilities0}
\begin{aligned}
u_{i,i\pm k} &= 0, \qquad k\ge 2, \\
u_{i,i+1} &= \frac{N-i}{N}\frac{i}{N}
\left(1+e^{-\beta[\Pi_C(i)-\Pi_D(i)]}\right)^{-1}, \\
u_{i,i-1} &= \frac{N-i}{N}\frac{i}{N}
\left(1+e^{\beta[\Pi_C(i)-\Pi_D(i)]}\right)^{-1}, \\
u_{i,i} &= 1-u_{i,i+1}-u_{i,i-1}.
\end{aligned}
\end{equation}


\subsection{Social dilemmas: Donation and Public Goods games}
Our analysis will be carried out for cooperation dilemmas in both pairwise and multi-player settings, described below.
\subsubsection{Donation Game (DG)}

\noindent The Donation Game is a special case of the Prisoners' Dilemma \cite{sigmund2010calculus}, where cooperation corresponds to providing the co\text{-}player with a benefit $b$ at a personal cost $c$, with $b>c$, while defection yields no benefit and incurs no cost. The payoff matrix of the game (for the row player) is given by
\[
 \bordermatrix{~ & C & D\cr
                  C & b-c & -c \cr
                  D & b & 0  \cr
                 }.
\]

\noindent Let $\pi_{X,Y}$ denote the payoff of a player using strategy $X \in \{C,D\}$ when interacting with a player using strategy $Y \in \{C,D\}$. In a well\text{-}mixed population of size $N$, at state $S_i$ where there are $i$ cooperators, the expected payoffs of a C\text{-}player and a D\text{-}player are given by
\begin{equation*}
\begin{aligned}
\Pi_C(i) &= \frac{(i-1)\pi_{C,C} + (N-i)\pi_{C,D}}{N-1}
         = \frac{(i-1)(b-c) + (N-i)(-c)}{N-1}, \\
\Pi_D(i) &= \frac{i\pi_{D,C} + (N-i-1)\pi_{D,D}}{N-1}
         = \frac{ib}{N-1}.
\end{aligned}
\end{equation*}

\noindent Therefore, the payoff difference between cooperation and defection is
\[
\delta = \Pi_C(i) - \Pi_D(i) = -\Big(c + \frac{b}{N-1}\Big),
\]
which is negative and independent of the population state $S_i$, in accordance with the general assumption introduced earlier.

\subsubsection{Public Goods Game (PGG)}
\noindent In the Public Goods Game, individuals interact in groups of size $n$ \cite{hauert2007via}. Each player can either cooperate by contributing an amount $c>0$ to a common pool, or defect by contributing nothing. The total contribution within a group is multiplied by an enhancement factor $r$, with $1<r<n$, and the resulting amount is equally shared among all group members, independently of their strategies. Since defectors benefit from the public good without paying the cost, the game constitutes a social dilemma.

\noindent In a well\text{-}mixed population of size $N$, at state $S_i$ where there are $i$ cooperators, groups are formed by multivariate hypergeometric sampling. Hence, the expected payoffs of a C\text{-}player and a D\text{-}player are given by
\begin{equation*}
\begin{aligned}
\Pi_C(i)
&= \sum^{n-1}_{k=0}
\frac{\dbinom{i-1}{k}\dbinom{N-i}{\,n-1-k\,}}{\dbinom{N-1}{\,n-1\,}}
\left(\frac{(k+1)rc}{n} - c \right)
= \frac{rc}{n}\left(1 + (i-1)\frac{n-1}{N-1}\right) - c, \\[6pt]
\Pi_D(i)
&= \sum^{n-1}_{k=0}
\frac{\dbinom{i}{k}\dbinom{N-1-i}{\,n-1-k\,}}{\dbinom{N-1}{\,n-1\,}}
\frac{krc}{n}
= \frac{rc(n-1)}{n(N-1)}\, i .
\end{aligned}
\end{equation*}

\noindent Therefore, the payoff difference between cooperation and defection is
\[
\delta = \Pi_C(i) - \Pi_D(i)
= -c \left(1 - \frac{r(N-n)}{n(N-1)} \right),
\]
which is negative and independent of the population state $S_i$.

\subsection{Cost optimisation under institutional incentives}

To reward a cooperator (respectively, punish a defector), the institution has to spend an amount $\theta$ (per-capita incentive), such that the payoff of the targeted individual increases  by $a \theta$ (respectively, a decrease by $\hat{a} \theta $), where $a$ and $\hat{a}$ denote the efficiencies of reward and punishment, respectively.


We derive the expected cost of providing institutional incentives \cite{han2018cost,duonghan2021cost}. 
Under institutional incentives, for $1\le i\le N-1$, the transition probabilities in \eqref{eq: transition probabilities0} are modified as follows for  reward (and similarly for punishment):
\begin{equation}
\label{eq: transition probabilities}
\begin{aligned}
u_{i,i\pm k} &= 0, \qquad k\ge 2, \\
u_{i,i+1} &= \frac{N-i}{N}\frac{i}{N}
\left(1+e^{-\beta[\Pi_C(i)-\Pi_D(i)+a \theta]}\right)^{-1}, \\
u_{i,i-1} &= \frac{N-i}{N}\frac{i}{N}
\left(1+e^{\beta[\Pi_C(i)-\Pi_D(i)+a \theta]}\right)^{-1}, \\
u_{i,i} &= 1-u_{i,i+1}-u_{i,i-1}.
\end{aligned}
\end{equation}

Let $\mathcal{N}=(I-U)^{-1}=(n_{ik})_{i,k=1}^{N-1}$ denote the fundamental matrix of this chain. The entry $n_{ik}$ gives the expected number of visits to state $S_k$ when starting from state $S_i$. Since mutants can appear with equal probability in $S_0$ and $S_N$, the expected number of visits to $S_j$ is,
$
\frac{1}{2}(n_{1j}+n_{N-1,j})
$. Thus, the expected costs of using only reward and only punishment are given by 
\begin{equation}
\label{eq: Er and Ep}
E_r(\theta)=\frac{\theta}{2}\sum_{j=1}^{N-1}(n_{1j}+n_{N-1,j})j,
\qquad
E_p(\theta)=\frac{\theta}{2}\sum_{j=1}^{N-1}(n_{1j}+n_{N-1,j})(N-j).
\end{equation}


Now, we derive the cooperation frequency under institutional incentives \cite{han2018cost}. Since the population consists of two strategies, the fixation probabilities of a single cooperator in a population of defectors and vice versa are given by 
\begin{equation*}
\begin{aligned}
\rho_{D,C}
&=\left(1+\sum_{i=1}^{N-1}
\prod_{k=1}^{i}
\frac{1+e^{\beta[\Pi_C(k)-\Pi_D(k)+a\theta]}}
{1+e^{-\beta[\Pi_C(k)-\Pi_D(k)+a\theta]}}
\right)^{-1}, \\
\rho_{C,D}
&=\left(1+\sum_{i=1}^{N-1}
\prod_{k=1}^{i}
\frac{1+e^{\beta[\Pi_D(k)-\Pi_C(k)-a\theta]}}
{1+e^{-\beta[\Pi_D(k)-\Pi_C(k)-a\theta]}}
\right)^{-1}.
\end{aligned}
\end{equation*}

As the stationary frequency of cooperation is given by
$\frac{\rho_{D,C}}{\rho_{D,C}+\rho_{C,D}}$, 
maximising this frequency is equivalent to maximising
\begin{equation}
\label{eq:max}
\max_{\theta}\left(\frac{\rho_{D,C}}{\rho_{C,D}}\right).
\end{equation}
This ratio simplifies as 
\begin{eqnarray}
\nonumber
\frac{\rho_{D,C}}{\rho_{C,D}}
&=&\prod_{k=1}^{N-1}
\frac{u_{k,k-1}}{u_{k,k+1}}
=\prod_{k=1}^{N-1}
\frac{1+e^{\beta[\Pi_C(k)-\Pi_D(k)+a\theta]}}
{1+e^{-\beta[\Pi_C(k)-\Pi_D(k)+a\theta]}} \\
\nonumber
&=& e^{\beta\sum_{k=1}^{N-1}(\Pi_C(k)-\Pi_D(k)+a\theta)} \\
\label{eq:max_Q_prime}
&=& e^{\beta(N-1)(\delta+a\theta)}
\end{eqnarray}

Given that a minimal   level of population cooperation $\omega\in[0,1]$ is required. That is, the following inequality must hold,
$
\frac{\rho_{D,C}}{\rho_{D,C}+\rho_{C,D}}\ge \omega
$. It follows from \eqref{eq:max_Q_prime} that
\begin{equation}
\label{eq:omega_fraction}
\theta \ge \theta_\omega
= \frac{1}{a(N-1)\beta}
\log\!\left(\frac{\omega}{1-\omega}\right)-\delta.
\end{equation}


\subsection{Social welfare optimisation under institutional incentives}
Below we derive the population social welfare under institutional reward and punishment. 

\subsubsection{Institutional reward} Recall that $a \in [0,+\infty)$  represents the efficiency of the reward mechanism, i.e. a (per-capita) institutional cost of $\theta$ yields a payoff increase of $a\theta$ for a targeted cooperator. The reward mechanism is deemed cost-efficient when $a \ge 1$. Our analysis examines the mathematical properties of social welfare under three distinct efficiency regimes: $a = 1$, $a < 1$, and $a > 1$.

We begin by deriving $SW(\theta)$, the expected total social welfare across all population states. For a specific state $S_i$ containing $i$ cooperators, social welfare is calculated as the total payoff of all players in that state (denoted by $P_i$) strictly net of the total institutional cost (denoted by $\theta_i$).


We define $\Delta:=\Pi_D(i)/i$, that is
\begin{equation*}
\Delta=\begin{cases}
    \frac{b}{N-1}\quad \text{in Donation Game},\\
    \frac{rc(n-1)}{n(N-1)}\quad \text{in Public Goods Game}.
\end{cases}
\end{equation*}
Note that both $\delta$ and $\Delta$ are independent of the states of the population. We obtain,
$P_i = i\big[\Pi_C(i) + a\theta\big] + (N - i)\,\Pi_D(i)$. 
As $\theta_i=i\theta$, the population social welfare in state $S_i$ is:
\begin{align*}
SW_i(\theta)&= P_i - \theta_i \notag \\
 &= i\big[\Pi_C(i) + a\theta\big] + (N - i)\,\Pi_D(i) - i\theta \notag \\
 &= i\Pi_C(i) + (N - i)\Pi_D(i) + i\theta(a-1) \notag \\
 &= i\big(\Pi_C(i) - \Pi_D(i)\big) + N\Pi_D(i) + i(a-1)\theta \notag \\
 &= i\delta + N(i\Delta) + i(a-1)\theta \notag \\
 &=i\big(\delta + N\Delta + (a-1)\theta \big).
\end{align*}
To evaluate the expected population-level social welfare, we aggregate $SW_i(\theta)$ over all transient states $S_i$ ($i = 1$ to $N-1$), weighted by the expected visitation rates ($\frac{n_{1,i} + n_{N-1,i}}{2}$) derived from the Fundamental Matrix ($\mathcal{N}$) defined in \cite{duonghan2021cost}. 
\begin{align}
SW(\theta)&=\frac{1}{2}\sum_i SW_i(\theta)(n_{1,i} + n_{N-1,i})\nonumber\\
&= \frac{1}{2}\sum_{i}i\big(\delta + N\Delta+(a-1)\theta \big) 
(n_{1,i} + n_{N-1,i})\nonumber\\ 
&= \frac{1}{2}(\delta + N\Delta + (a-1)\theta) 
\sum_{i}i (n_{1,i} + n_{N-1,i})\nonumber\\
&= \frac{N^2}{2}\frac{f(x)}{g(x)} \big(\delta + N\Delta + (a-1)\theta\big),
 \label{eq:SW-theta-final}
\end{align}
where $x:=\beta(a\theta + \delta)$, $f(x)$ and $g(x)$ are two functions defined below \cite{duonghan2021cost}: 
\begin{align}
f(x) &= (1 + e^x)\left[ \left(1 + e^x + \cdots + e^{(N-2)x}\right) H_N 
+ e^{(N-1)x} \sum_{j=1}^{N-1} \frac{e^{-jx}}{j} \right] \label{fx}, \\
g(x) &= 1 + e^x + \cdots + e^{(N-1)x}. \label{gx}
\end{align}


The main objective of this paper is to study the mathematical problem of optimising the total expected social welfare (under institutional reward)
\begin{equation}
 \label{eq: max SW} \max_{\theta>0}SW(\theta).  
\end{equation}
\subsubsection{Institutional punishment}
For punishment, the institution incurs a cost $\theta$, resulting in a reduction of $\hat{a}\theta$ to the defector's payoff. In this case, $
\hat{P}_i = i\,\Pi_C(i) + (N - i)\big[ \Pi_D(i) - \hat{a}\theta \big].$ Combining with the total institutional cost $\hat{\theta}_i = (N - i)\theta$ yields aggregate the social welfare in state $S_i$ as:
\begin{equation*}
\widehat{SW}_i(\theta) = \hat{P}_i - \hat{\theta}_i =  i\,\Pi_C(i) + (N - i)\big[ \Pi_D(i) - (1+\hat{a})\theta \big]
\end{equation*}
By similar computations as in the reward case, we obtain the following expression for the total expected social welfare in the punishment incentive (see Appendix \ref{sec: punishment calculations} for the  detailed calculations)
\begin{equation}
\widehat{SW}(\theta)= \frac{N^2}{2}\frac{f(x)}{g(x)} \big(\delta + N\Delta) - \frac{N^2}{2}\frac{\hat{f}(x)}{g(x)}(1+\hat{a})\theta, \label{eq:swp0}   
\end{equation}
where
\[
\hat{f}(x) = (1 + e^x)\left[ \left(1 + e^x + \cdots + e^{(N-2)x}\right) H_N 
+ \sum_{j=1}^{N-1} \frac{e^{(j-1)x}}{j} \right].
\]
The corresponding social welfare optimisation problem (under institutional punishment) is  defined  as
\begin{equation} 
\label{eq: max SWp}\max_{\theta>0}\widehat{SW}(\theta).
\end{equation}

\section{Main results}
\label{sec: main results}
The aim of this paper is to provide a rigorous analysis of the total expected social welfare $SW$ and $\widehat{SW}$ and the associated optimisation problems \eqref{eq: max SW} and \eqref{eq: max SWp}. From a mathematical point of view, these are  non-trivial optimisation problems since the objective functions
are generally non-convex and depend on several parameters, namely the population and group sizes $N$
and $n$ (which can be arbitrarily large), the payoff entries, the strength of selection as well as  the efficiency of
the institutional incentives. We obtain both analytical results, analysing qualitative properties of the social welfare functions, and numerical results, offering algorithms to practically compute the optimal institutional cost.

\subsection{Reward}
The first result of our paper is the following theorem, which shows that both the efficiency of the reward mechanism ($a$) and the strength of selection ($\beta$) exert non-trivial effects on the qualitative behaviour of $SW(\theta)$. When $a=1$ (zero-sum transfer), $SW$ is increasing and then decreasing for all $\beta$. However, when $a\neq 1$ (non-zero sum transfer), the monotonicity of $SW$ undergoes an intriguing phase-transition phenomenon as $a$ and $\beta$ vary. We provide analytical formula for computing the critical thresholds and the optimiser in each scenario.
\begin{theorem}
(Behaviour and optimisation of the total expected social welfare)
\label{thm: qualitative SW}

\textbf{(1) Zero-sum transfer (${a=1}$)} When $a=1$, there exists a unique $\theta^*>0$ such that $SW(\theta)$ is increasing on $(0,\theta^*)$ and decreasing on $(\theta^*,+\infty)$. Thus, $SW$ has a unique global maximiser at $\theta=\theta^*$.  

\textbf{(2) Efficient institutional reward ($a>1$)}. We define a threshold:    \begin{equation}
            \beta^* = -\frac{F^*}{\mathcal{K}}>0 \label{beta_star}.
        \end{equation}
   
\begin{enumerate}
\item[(i)] For $\beta \le \beta^*$, there exists a threshold $\theta_0 > 0$ such that $SW(\theta)$ is non-decreasing on $(\theta_0, +\infty)$.
\item[(ii)] \textit{(behaviour above the threshold value)} For $\beta > \beta^*$, the number of changes of the sign of $dSW(\theta)/d\theta$ is at least two for all $N$ and there exists an $N_0$ such that the number of changes is exactly two for $N \le N_0$. As a consequence, for $N \le N_0$, there exists $\theta_1 < \theta_2$ such that, for $\beta > \beta^*$, $SW(\theta)$ is increasing when $\theta < \theta_1$, decreasing when $\theta_1 < \theta < \theta_2$ and increasing when $\theta > \theta_2$. 
\end{enumerate}
\textbf{(3) Inefficient institutional reward ($a<1$)}
\begin{enumerate}
        \item[(i)] There exist a threshold $a^* \in (0, 1)$ such that $SW(\theta)$ is strictly decreasing on $(\theta_0, +\infty)$ when $a^*<a<1$.
        \item[(ii)] \textit{(behaviour under the threshold value)} For $0<a<a^*$ and $\beta < \beta^*$, $SW(\theta)$ is non-decreasing on $(\theta_0, +\infty)$. Consequently:
        \[
        \max_{\theta \ge \theta_0} SW(\theta)=SW(\theta_0)
        \]

        \item[(iii)] \textit{(behaviour above the threshold value)} For $0<a<a^*$ and $\beta > \beta^*$, the number of changes of the sign of $dSW(\theta)/d\theta$ is at least two for all $N$ and there exists an $N_0$ such that the number of changes is exactly two for $N \le N_0$. As a consequence, for $N \le N_0$, there exist $\theta_1 < \theta_2$ such that, for $\beta > \beta^*$, $SW(\theta)$ is decreasing when $\theta < \theta_1$, increasing when $\theta_1 < \theta < \theta_2$ and decreasing when $\theta > \theta_2$. Thus, for $N \le N_0$:
        \[
        \max_{\theta \ge \theta_0}SW(\theta)=\max\{SW(\theta_0),SW(\theta_2)\}  
        \]

        \item[(iv)] Moreover, for sufficiently large $\beta$ and small $\theta$, $SW(\theta)$ is increasing as $\theta \to 0^+$ (see proof in  Lemma \ref{lem:lower_bound_c} in Appendix).
    \end{enumerate}
\end{theorem}
Plots of the qualitative behaviour of the social welfare $SW$ as a function of $\theta$ for various parameters as well as numerical calculations of the critical thresholds $\theta^*$ and $\beta^*$ are presented in Figure \ref{fig:thm1-reward-dg} for Donation Game and Figure \ref{fig:thm1-reward-pgg} for Public Goods Games.

Our second result is the following theorem which shows that the maximiser, $\theta^*$, of the optimisation problem \eqref{eq: max SW} is either zero or located around a specific value, $\theta_\infty=-\frac{\delta}{a}$.
\begin{theorem}[localsation around $\theta_\infty$]
\label{thm:localization}
Let $\theta^\star$ be a maximiser of the social welfare objective $SW(\theta)$ over $\theta \ge 0$, and define
\[
\theta_\infty = -\frac{\delta}{a}.
\]
Then either $\theta^\star = 0$, or $\theta^\star>0$ and
\[
\bigl|\theta^\star - \theta_\infty\bigr| = \mathcal{O}\!\left(\frac{1}{a\beta}\right).
\]
\end{theorem}


Based on Theorem \ref{thm:localization}, we develop the following algorithm (Algorithm 1) to approximate the optimal incentive $\theta^*$. It is applicable to both the Donation Game and the Public Goods Game described.

\begin{algorithm}[H]
\caption{Restricted Interval Grid Search with Upper and Lower Bounds Cutoff}
\label{alg:pgg_main}
\KwIn{Game parameters $b, c$; Parameters $N \in \mathbb{Z}^{+}$, $a, \beta \in \mathbb{R}^{+}$ ($a < 1$). Search radius $r > 0$ and number of steps $N_{steps} \in \mathbb{Z}^{+}$}
\KwOut{$\theta^{*}$ maximising $SW$ within valid bounds}

\textbf{Step 1: Compute Bound and Starting Point}

$\;\;\;\;\;\; \text{Compute values of } \delta \text{ and } \Delta$

$\;\;\;\;\;\;\theta_{limit} \gets \frac{\delta + N\Delta}{1-a}$ \tcp*[r]{Define the hard upper limit}

$\;\;\;\;\;\;\theta_{start} \gets \min \left( \frac{-\delta}{a}, \theta_{limit} \right)$;

\textbf{Step 2: Initialize Grid Search}

$\;\;\;\;\;\;\theta_{best} \gets 0$;

$\;\;\;\;\;\;SW_{best} \gets SW(0)$;

$\;\;\;\;\;\;\Delta \theta \gets 2r / N_{steps}$ \tcp*[r]{Calculate step size}

\textbf{Step 3: Evaluate Interval}

\For{$k \gets 0$ \textbf{to} $N_{steps}$}{
    $\theta_{curr} \gets \max(0, (\theta_{start} - r) + k \cdot \Delta\theta);$

    \If{$\theta_{curr} > \theta_{limit}$}{
        \textbf{break} \tcp*[r]{Terminate search if limit is exceeded}
    }

    $val \gets SW(\theta_{curr})$;

    \If{$val > SW_{best}$}{
        $SW_{best} \gets val$;

        $\theta_{best} \gets \theta_{curr}$;
    }
}

\textbf{Step 4: Return Optimal Incentive}

\Return{$\theta_{best}$};
\end{algorithm}

Theorem \ref{thm: qualitative SW} has shown the influence of the strength of selection to the total expected social welfare. The following theorem further characterises the asymptotic behaviour of $SW$ in the neutral and strong selection limits, namely when $\beta$ tends to zero and infinity respectively.

The convergence of the optimiser $\theta^*$ to $\theta_\infty$ is numerically presented in Figure \ref{fig:thm2_4_convergence}.\\

\begin{theorem}(Neutral and strong selection limits) \label{thm: selection-limits}

$SW(\theta)$ exhibits linear behaviour in the limits $\beta \to 0^+$ and $\beta \to +\infty$. More precisely, we have
\begin{align*}
\lim_{\beta\to 0^+} SW(\theta)&=N^2 H_N (\delta+N\Delta+(a-1)\theta),\\
\lim_{\beta \to +\infty}SW(\theta)&=\begin{cases}
      N^2 H_N (\delta+N\Delta+(a-1)\theta)\quad\text{for}\quad \theta=-\frac{\delta}{a},\\
      \frac{N^2}{2}(H_N+1)(\delta+N\Delta+(a-1)\theta)\quad\text{for}\quad \theta>-\frac{\delta}{a},\\
      \frac{N^2}{2}\Big(H_N+\frac{1}{N-1}\Big) (\delta+N\Delta+(a-1)\theta)\quad\text{for}\quad \theta<-\frac{\delta}{a}.
    \end{cases}     
\end{align*}
In the above formulae, $H_N$ denotes the harmonic number $H_N=\sum_{j=1}^{N-1}\frac{1}{j}.$
\end{theorem}
In Figures \ref{fig:thm3-reward-dg} and \ref{fig:thm3-reward-pgg} we numerically demonstrate the asymptotic limits of the social welfare described in Theorem \ref{thm: selection-limits} respectively for Donation Game and Public Goods Game.

\subsection{Institutional punishment vs institutional reward}
We now consider institutional punishment \footnote{A counterpart of Theorem \ref{thm: qualitative SW} for the punishment could be obtained, but we omit it and only present Theorem \ref{thm:localization_p} since it provides an algorithm to compute the optimiser in practice.}.
The following theorem is the counterpart of Theorem \ref{thm:localization} for this type of incentive. It shows the localisation property of the optimal of $SW(\theta)$ for the punishment.
\begin{theorem}
\label{thm:localization_p}
Let $\theta^\star$ be a maximiser of the social welfare function  $\widehat{SW}(\theta)$ over $\theta \ge 0$, and define
\[
\hat{\theta}_\infty = -\frac{\delta}{\hat{a}}.
\]
Then either $\theta^\star = 0$, or $\theta^\star>0$ and
\[
\bigl|\theta^\star - \hat{\theta}_\infty\bigr| = \mathcal{O}\!\left(\frac{\log \beta}{\hat{a}\beta}\right).
\]
\end{theorem}
Compared with the convergence rate for institutional  reward  obtained in Theorem \ref{thm:localization}, the convergence for the punishment one is  slower because of the extra logarithmic factor.

The following theorem shows that if the reward efficiency is sufficiently high compared to that of the punishment, then the institution can always adjust the investment cost so that the reward incentive outperforms the punishment one.\\

\begin{theorem} \label{thm:eff}
Let $a$ and $\hat{a}$ be the efficiencies of  institutional reward and punishment, respectively. For $N \ge 3$, it satisfies that $\widehat{SW}(\theta) \le SW(\hat{a}\theta/a)$ for all $\theta \ge 0$ if and only if 
\[
a \ge \frac{\hat{a} \eta_{N-1}}{\eta_0 + \hat{a}(\eta_0 + \eta_{N-1})},
\]
where 
\[
\eta_0 = \frac{1}{N-1} + H_{N}\quad \text{and}\quad
    \eta_{N-1} = 1 + H_{N},
\]   
Specifically, when the two types of institutional incentives are equally efficient, i.e. $a = \hat{a}$, reward leads to a higher level of social welfare than punishment whenever 
\[
a \ge \frac{\eta_{N-1}-\eta_0}{\eta_{N-1}+\eta_0} = \frac{N-2}{N + 2(N-1)H_N}.
\]
\end{theorem}
Comparisons between the reward and punishment incentives are numerically demonstrated in Figure \ref{fig:thm5-dynamics}.
\subsection*{Idea of the proofs} 
The technically detailed proofs of the above theorems are deferred to Appendices  \ref{sec: appendix1} and \ref{sec: appendix2}. Here we provide the underlying ideas of these proofs.

The proofs of Theorems \ref{thm: qualitative SW}-\ref{thm: selection-limits}-\ref{thm:eff} rely on the analytically explicit formulas of the objective functions, which are the total expected social welfare $SW$ and $\widehat{SW}$ given in \eqref{eq:SW-theta-final}-\eqref{eq:swp0}. We are able to derive explicitly these formulas because of the crucial fact mentioned earlier that in Donation Game and Public Goods Game the payoff difference between cooperation and defection is independent of the population states. Although these expressions are still complicated, they enable us to compute the derivative of the objective functions. Finding the roots of the derivative functions then boils down to finding the roots of certain polynomials, see Equations \eqref{eq: def of P} and \eqref{eq: def F} below, which has been studied in \cite{duonghan2021cost}. The qualitative behaviour, especially the monotonic properties, of the objective functions is then followed by a thorough analysis of the sign of the derivative functions.

The later step is also the key step in the proofs of Theorems \ref{thm:localization}-\ref{thm:localization_p}. By characterising the sign of the derivative functions in appropriate intervals, we are able to provide lower and upper estimates for the location of their roots, which are the optimisers of the corresponding optimisation problems. 
\section{Numerical analyses and validation}
\label{sec: numerics}
This section analyses the behaviour of the expected social welfare ($SW$) with respect to the incentive impact $\theta$, under different incentive regimes and game  parameters. To further verify the analytical results (Theorems 1-5) presented above, we focus on:
\begin{enumerate}
    \item The effect of efficiency parameter $a$ and selection intensity $\beta$ on the variation of the $SW$ curve for the reward case, notably on $\left[\theta_0, +\infty\right)$.
    \item Approximate linearity of $SW(\theta)$  under neutral and strong selection limits.
    \item Convergence of the optimal value of $\theta$, i.e. $\theta^*$, with respect to selection intensity $\beta$.
    \item Comparison between the social welfare functions, i.e. $SW(\theta)$ and $\widehat{SW}(\theta)$, under institutional reward and punishment policies.
    \item Comparison of the optimal incentive levels $\theta^*$ obtained under social welfare maximisation, institutional cost minimisation, and minimum frequency of cooperation constraints.  
\end{enumerate}

\subsection{Impact of varying incentive efficiency  $a$ and selection intensity $\beta$}

\begin{figure}[tbp]
    \centering
    \includegraphics[width=1\linewidth]{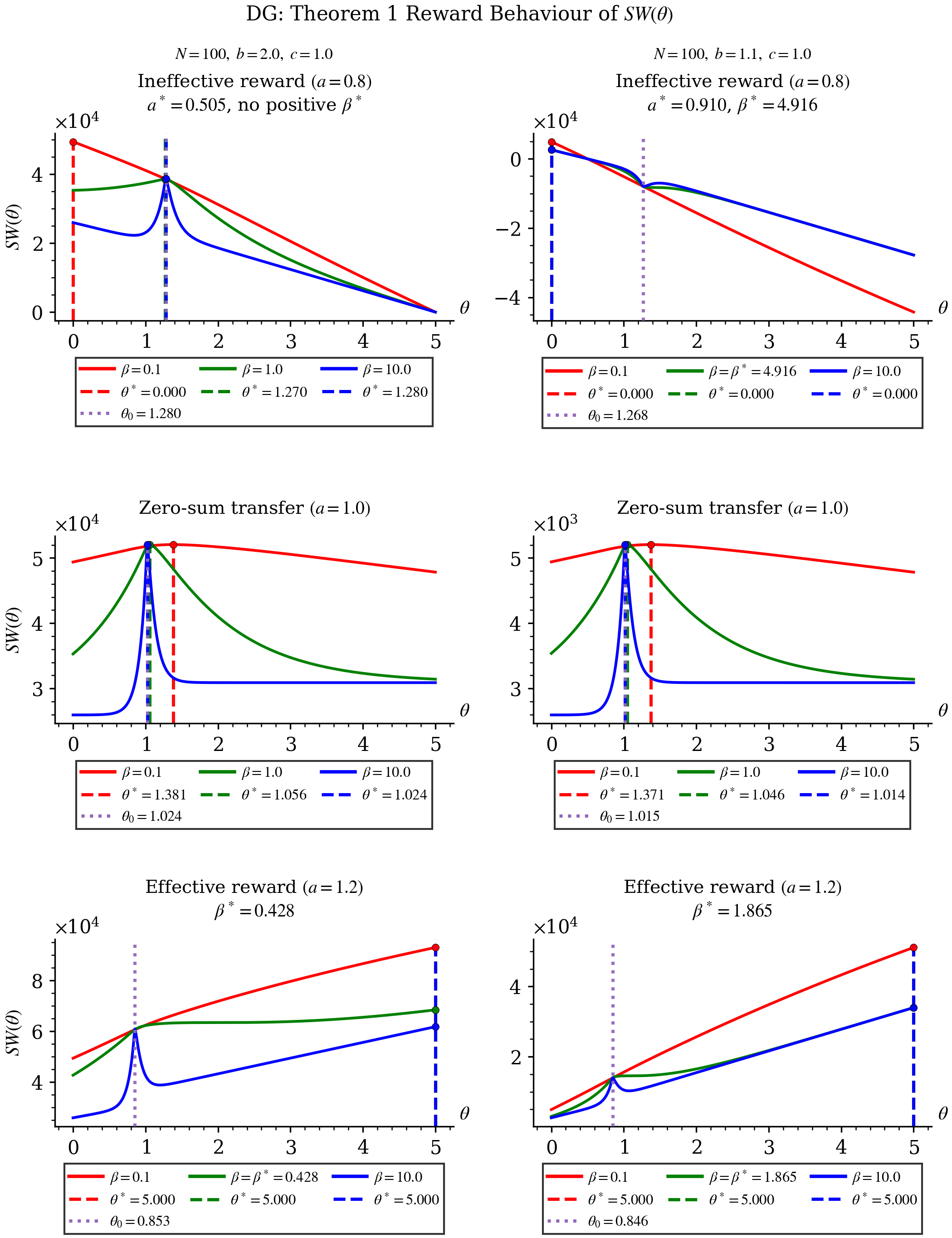}
    \caption{
    \textbf{In the Donation Game, depending on the selection intensity, the relationship between social welfare and the institutional incentive transitions is either monotonic behaviour or exhibits a clear extremum.}
    Social welfare $SW(\theta)$ as a function of the per-capita institutional cost $\theta$, for reward in the Donation Game (DG). The shape of the social welfare function undergoes qualitative transitions with incentive efficiency ($a$) and selection intensity ($\beta$). As $a$ increases, $SW$ changes from predominantly decreasing ($a<1$), to nearly flat ($a=1$), and eventually to increasing ($a>1$). Increasing $\beta$ further reveals a threshold $\beta^*$: below $\beta^*$, $SW$ is monotonic, whereas above $\beta^*$ the curve develops additional extrema, indicating a phase transition in welfare-maximising incentive levels.
    }
    \label{fig:thm1-reward-dg}
\end{figure}

\begin{figure}[tbp]
    \centering
    \includegraphics[width=1\linewidth]{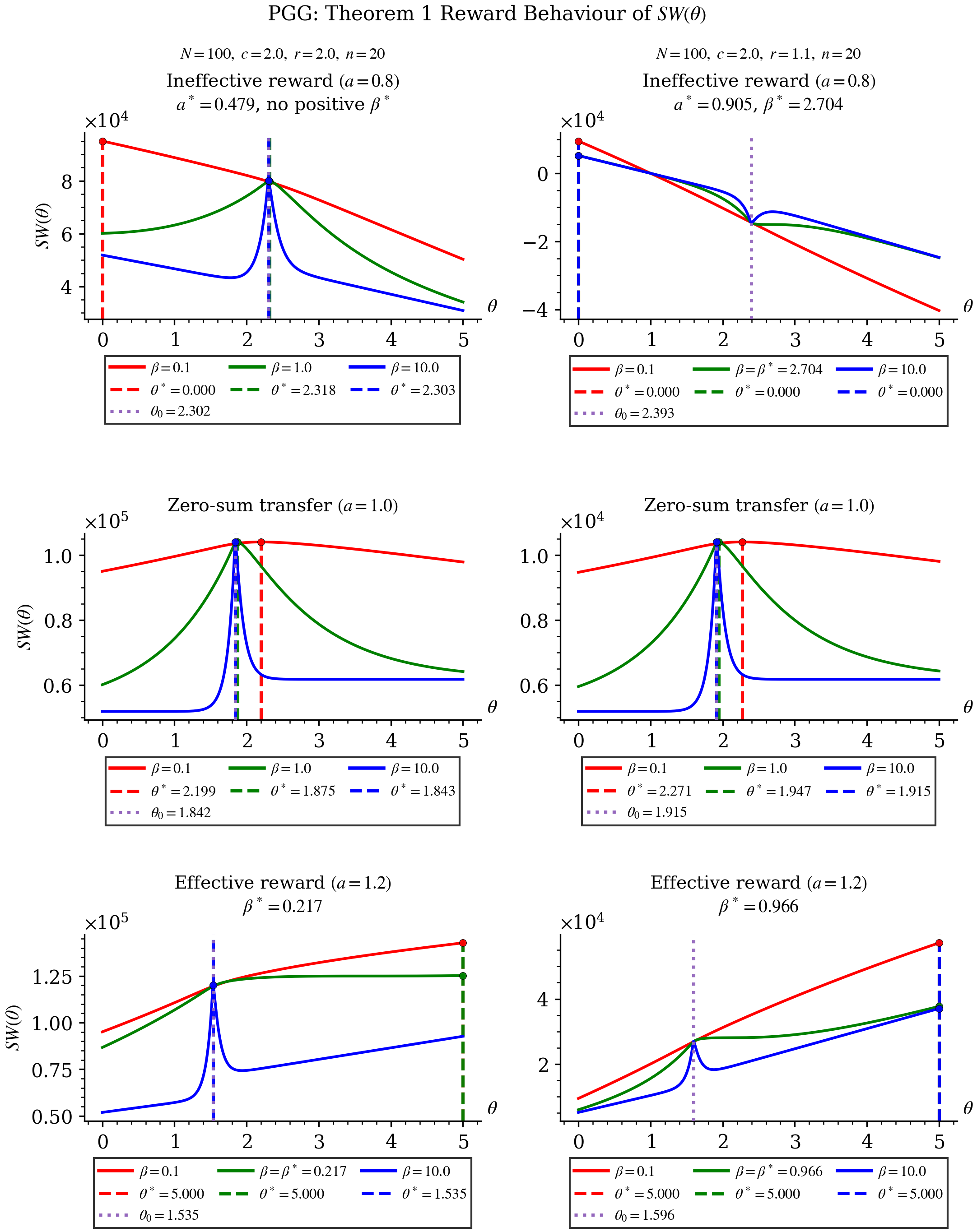}
    \caption{
    \textbf{In the Public Goods Game, depending on the selection intensity, the relationship between social welfare and the institutional incentive transitions is either monotonic behaviour or exhibits a clear extremum.}
    Shown are the numerical results for the Public Goods Game (PGG). The figures demonstrate the changes in the overall tendency of the $SW$ curve with varying efficiency parameter (from downward-sloping for $a<1$, to nearly level at $a=1$ and eventually upward-sloping when $a>1$), and the behaviour of the curve around threshold $\beta^*$.}
    \label{fig:thm1-reward-pgg}
\end{figure}

Increasing $a$ progressively reshapes the $SW$ curves, changing their overall tendency with respect to $\theta$ from decreasing to nearly flat and eventually increasing. This reflects a transition from a dissipative regime ($a<1$), where incentives reduce the overall population welfare, to an amplifying regime ($a>1$), where incentives enhance this outcome.

Increasing $\beta$ sharpens the structure of the $SW$ curve and changes its behaviour. As shown in Figures \ref{fig:thm1-reward-dg} and \ref{fig:thm1-reward-pgg}, when  $\beta$ is below the threshold defined in Theorem 1 (i.e. $\beta < \beta^*$), it  results in a monotonic $SW$ curve on $[\theta_0,+\infty)$. For $\beta > \beta^*$, the local extrema at $\theta_0$ can be observed clearly. The $SW$ function develops a distinct local extrema beyond $\theta_0$, indicating a phase transition. With $\beta \approx \beta^*$, it is unclear whether the $SW$ will exhibit a phase transition or a monotonic behaviour for all cases. This suggests that the change in the curve’s behaviour as $\beta$ approaches and surpasses the threshold is gradual and continuous, rather than abrupt.

The resulting observations in the behaviour of the $SW$ curve with respect to $a$ and $\beta$ further supports the analytical findings from Theorem \ref{thm: qualitative SW}.

\subsection{Linearity under neutral and strong selection limits}
\begin{figure}[tbp]
    \centering
    \includegraphics[width=1\linewidth]{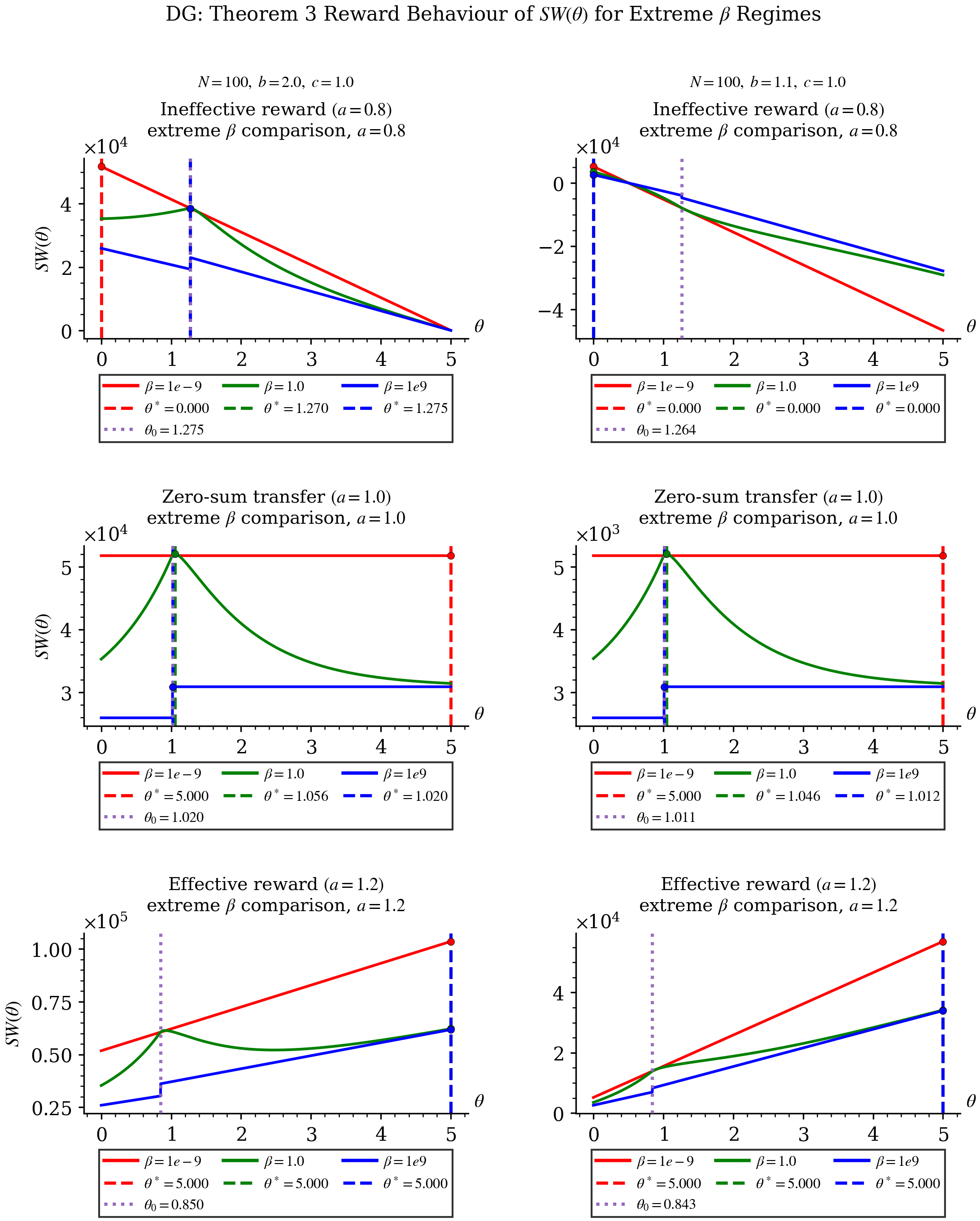}
    \caption{
    \textbf{In the Donation Game, social welfare undergoes a sharp phase transition at a critical incentive threshold under extreme selection intensities.}
    Shown are the numerical results for the Donation Game (DG). The figures illustrate how the $SW$ curve approaches a near-linear form under regimes $\beta \to 0^+$ and $\beta \to +\infty$, compared with its behaviour at an intermediate selection intensity.}
    \label{fig:thm3-reward-dg}
\end{figure}

\begin{figure}[tbp]
    \centering
    \includegraphics[width=1\linewidth]{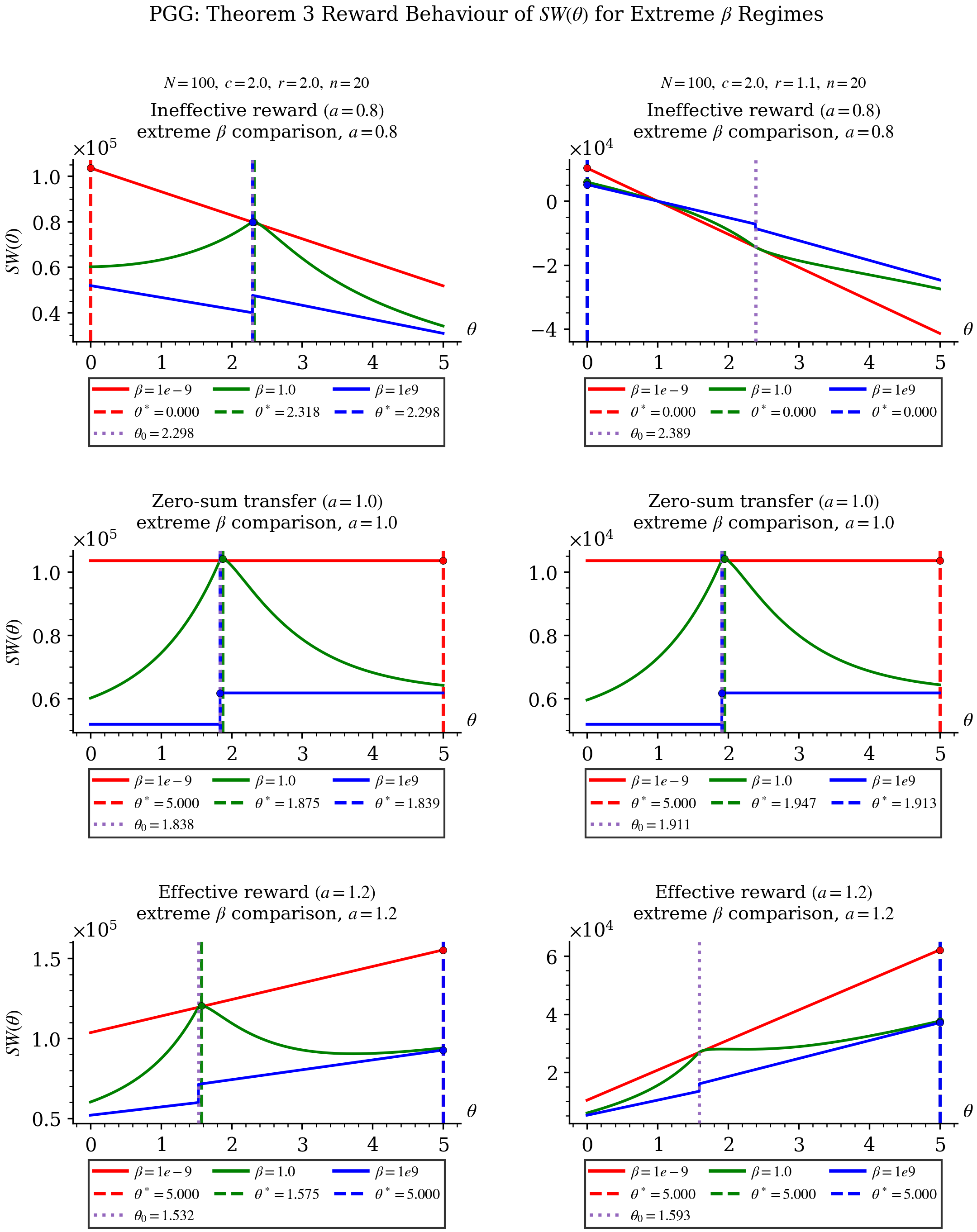}
    \caption{
    \textbf{In the Public Goods Game, social welfare undergoes a sharp phase transition at a critical incentive threshold under extreme selection intensities.}
    Shown are the numerical results for the Public Goods Game (PGG). The figures illustrate how the $SW$ curve approaches a near-linear form under regimes $\beta \to 0^+$ and $\beta \to +\infty$, compared with its behaviour at an intermediate selection intensity.}
    \label{fig:thm3-reward-pgg}
\end{figure}

The $SW$ behaviour observed in Figures \ref{fig:thm3-reward-dg} and \ref{fig:thm3-reward-pgg} are consistent with Theorem \ref{thm: selection-limits}. In the limiting regimes $\beta \to 0^+$ and $\beta \to +\infty$, $SW$ approaches a piecewise-linear profile, with approximately linear behaviour on $[0,\theta_0)$ and $(\theta_0,+\infty)$. However, it exhibits a steep, near-vertical transition around $\theta_0$ in several cases.

\subsection{Convergence of the optimal incentive}

\begin{figure}[htp]
    \centering
    \includegraphics[width=\linewidth]{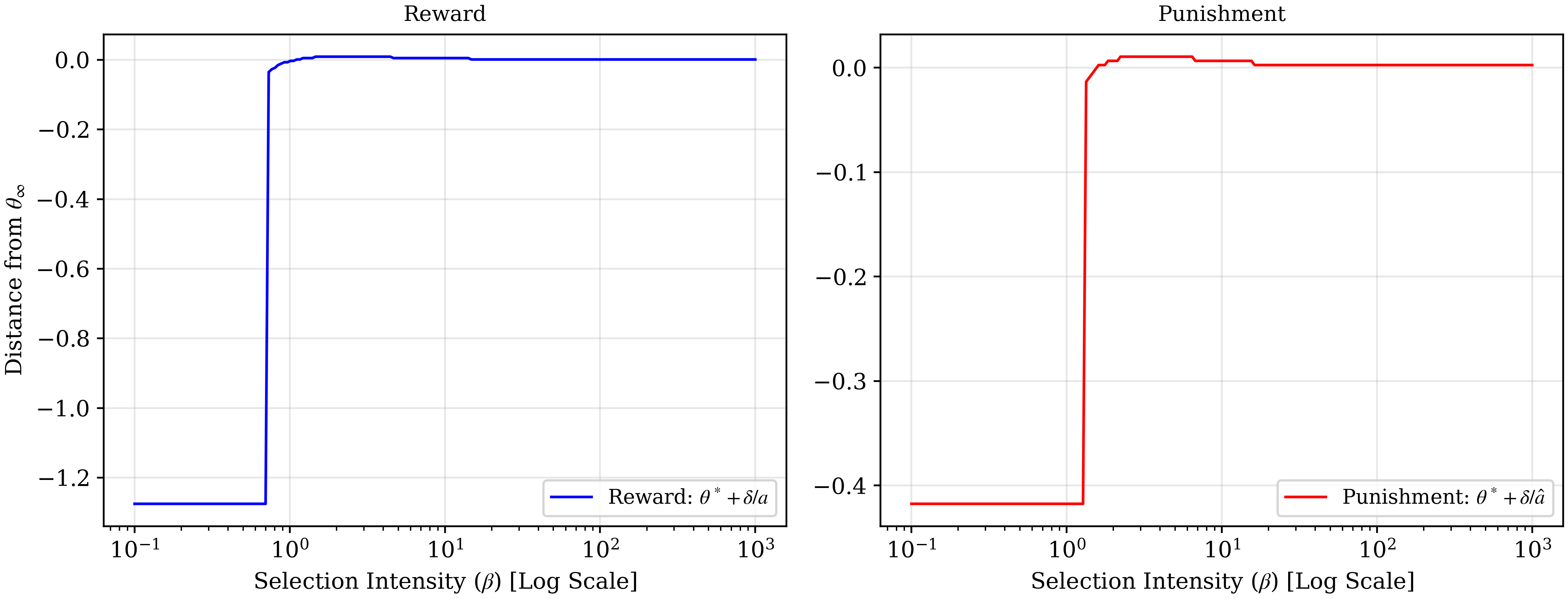}
    \caption{
    \textbf{The optimal institutional incentive undergoes a sharp phase transition at a critical selection threshold, rapidly converging to its theoretical limit under strong selection.}
    Convergence of the optimal incentive $\theta^*$ to theoretical limits as selection intensity $\beta$ increases on a logarithmic scale. (Left) Theorem \ref{thm:localization}: Optimal reward incentive converging to $\theta_\infty$. Simulated using the Donation Game (DG) with population size $N=100$, benefit $b=2.0$, cost $c=1.0$, and a reward transfer $a=0.8$. (Right) Theorem \ref{thm:localization_p}: Optimal punishment incentive converging to $\hat{\theta}_\infty$. Simulated using the DG with $N=100$, benefit $b=5.0$, cost $c=0.2$, and a punishment efficiency $\hat{a}=0.6$.}
    \label{fig:thm2_4_convergence}
\end{figure}

To validate the localisation properties of the optimal incentive $\theta^*$, we evaluate the distance between the numerical maximiser and the theoretical limits: $\theta_\infty$ for reward and $\hat{\theta}_\infty$ for punishment, as the selection intensity $\beta$ increases. Figure \ref{fig:thm2_4_convergence} illustrates this relationship for both mechanisms. Notably, the convergence dynamics exhibit a distinct phase transition: under weak selection, the optimal incentive $\theta^*$ remains at $0$, but once a critical threshold $\beta^*$ is exceeded, $\theta^*$ jumps away from zero and rapidly approaches its theoretical limiting value.

For institutional reward, Theorem \ref{thm:localization} states that the optimal incentive $\theta^*$ satisfies $|\theta^* - \theta_\infty| = \mathcal{O}(\frac{1}{a\beta})$. The left panel of Figure \ref{fig:thm2_4_convergence} visually confirms this bound for the Donation Game ($N=100, b=2.0, c=1.0, a=0.8$), demonstrating that as $\beta \to +\infty$, the distance between $\theta^*$ and $\theta_\infty$ collapses to zero. 

Similarly, for the punishment mechanism, Theorem \ref{thm:localization_p} establishes the bound $|\theta^* - \hat{\theta}_\infty| = \mathcal{O}(\frac{\log \beta}{\hat{a}\beta})$. The right panel validates this constraint using a modified parameter configuration ($N=100, b=5.0, c=0.2, \hat{a}=0.6$), showing a matching pattern of convergence. It is critical to note that for the punishment mechanism, achieving a non-trivial optimal incentive (where $\theta^* > 0$) is highly sensitive to the game's payoff structure. In our observations, this typically occurs only when the benefit-to-cost ratio ($b/c$) is exceptionally large. We utilised a ratio of $b/c = 25$ in this simulation specifically to capture this dynamic. Under less extreme conditions, the optimal punishment incentive frequently collapses to zero, underscoring the limitations of punishment compared to reward.

Regarding the convergence speeds, the theoretical bounds provide vital context for interpreting the numerical results. Theorem \ref{thm:localization} dictates that the reward mechanism converges at a rate of $\mathcal{O}(1/(a\beta))$, while Theorem \ref{thm:localization_p} establishes a slightly slower theoretical convergence rate for punishment at $\mathcal{O}(\log \beta / (\hat{a}\beta))$. Despite this mathematical distinction, the logarithmic scale of the x-axes in both figures reveals that once the phase transition threshold is crossed, both incentive mechanisms exhibit a rapid collapse toward their respective limits. 

Consequently, from an applied institutional perspective, both optimal incentives stabilise almost immediately once selection intensity becomes sufficiently strong.

\subsection{Social welfare dynamics under reward vs punishment}
\begin{figure}[htp]
    \centering
    \includegraphics[width=1\linewidth]{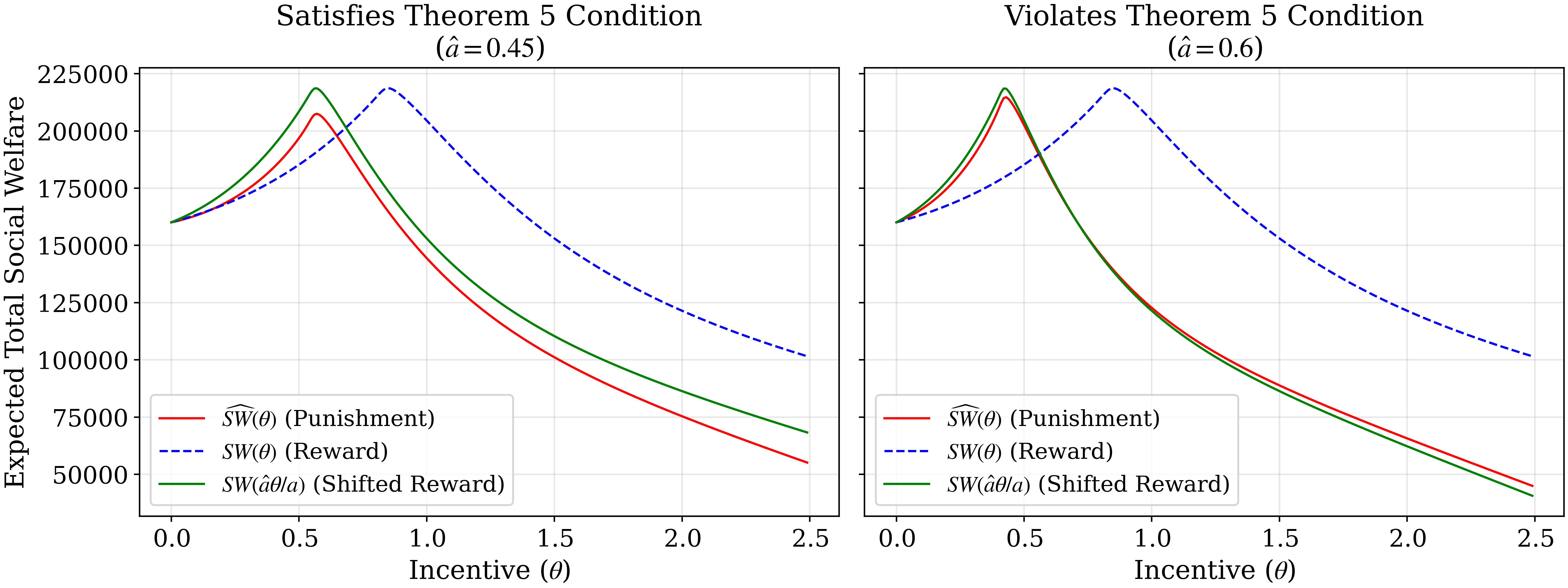}
    \caption{
    \textbf{The efficiency threshold dictates the absolute dominance of institutional reward over punishment.}
    Comparison of reward and punishment in the Donation Game ($N=100$, $b=5.0$, $c=0.2$, $a=0.3$). The left panel ($\hat{a}=0.45$) shows strict dominance of the shifted reward when the theoretical efficiency condition is met (see Theorem 5), while the right panel ($\hat{a}=0.6$) shows punishment outperforming reward when the condition is violated.}
    \label{fig:thm5-dynamics}
\end{figure}

To properly observe the regime where punishment could theoretically dominate, we utilize the same extreme payoff configuration from our convergence analysis: the Donation Game with a population size $N=100$, benefit $b=5.0$, and cost $c=0.2$. This very high benefit-to-cost ratio ($b/c = 25$) is necessary because, under standard payoff conditions, the reward mechanism almost universally outperforms in terms of social welfare maximisation.

In the left panel ($\hat{a} = 0.45$), the parameters satisfy the condition established in Theorem 5. Consequently, the shifted reward curve, $SW(\hat{a}\theta/a)$ (shown in green), acts as a strict upper bound to the punishment curve ($\widehat{SW}(\theta)$, shown in red) across all incentive levels. Not only does reward dominate, but the absolute gap in total social welfare between the two mechanisms remains visible even as the incentive $\theta$ increases. 

Conversely, the right panel illustrates a regime where this efficiency condition is explicitly violated by increasing the punishment efficiency to $\hat{a} = 0.6$. Here, the red punishment curve eclipses the green shifted reward curve for intermediate incentive values, demonstrating that the reward dominance fails to hold under these specific parameters.

\subsection{Social welfare vs institutional cost optimisation: optimal incentives}
\begin{figure}[thp]
    \centering
    \includegraphics[width=1\linewidth]{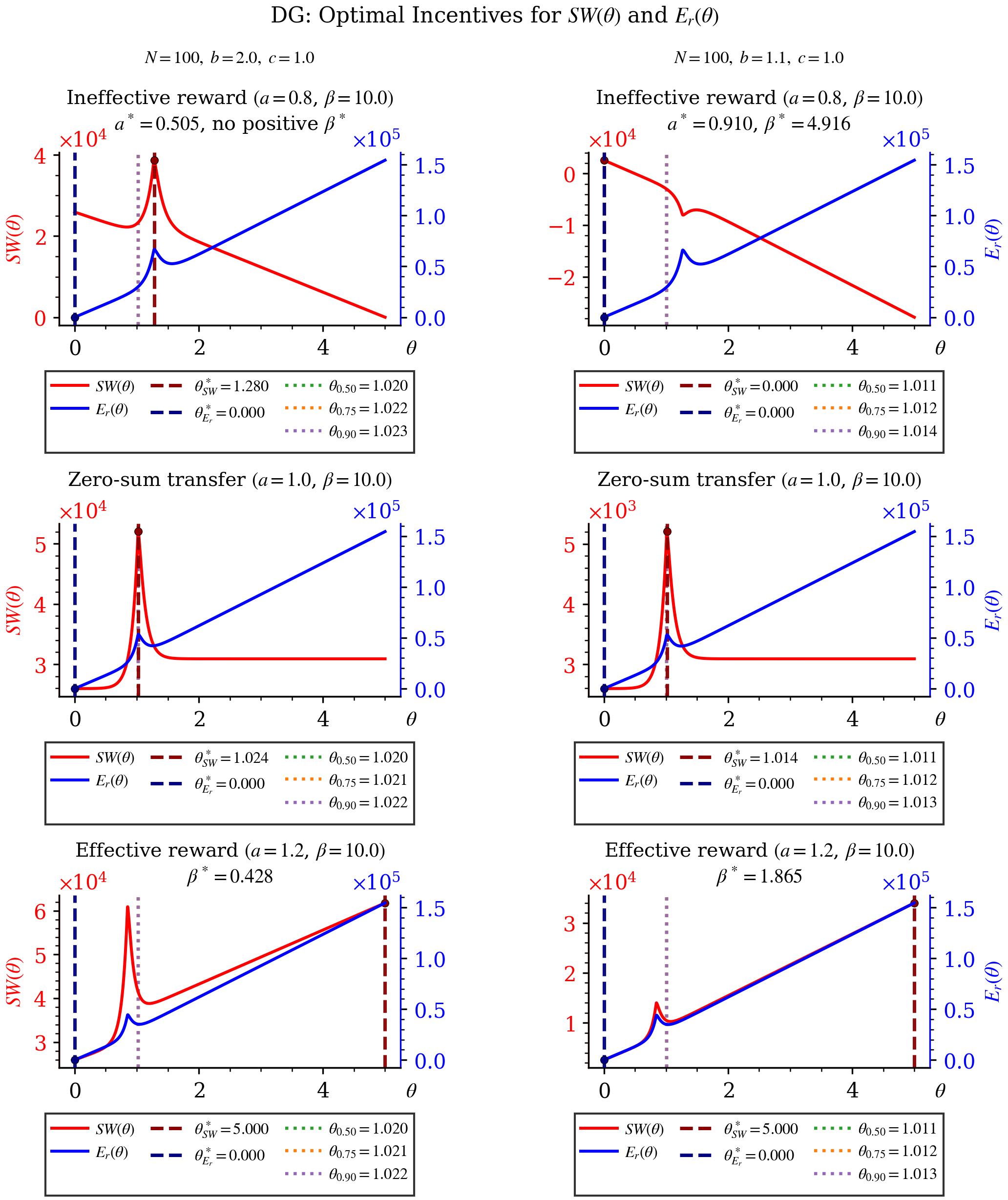}
    \caption{\textbf{Efficient rewards can conflict with cost: multi-objective comparison of optimal incentive levels for the DG (reward case, $\beta=10.0$).} The panels show that the social-welfare–maximising incentive often differs substantially from the cost-minimising incentive assuming  minimum cooperation targets. The figures illustrate how the incentive values that maximise $SW(\theta)$, minimise $E_r(\theta)$, and satisfy cooperation-frequency thresholds vary across three reward-efficiency regimes—ineffective reward ($a<1$), zero-sum transfer ($a=1$), and effective reward ($a>1$). Vertical lines mark the optimal incentives and the threshold values for target cooperation frequencies.}
    \label{fig:SW-vs-Er-DG}
\end{figure}

\begin{figure}[thp]
    \centering
    \includegraphics[width=1\linewidth]{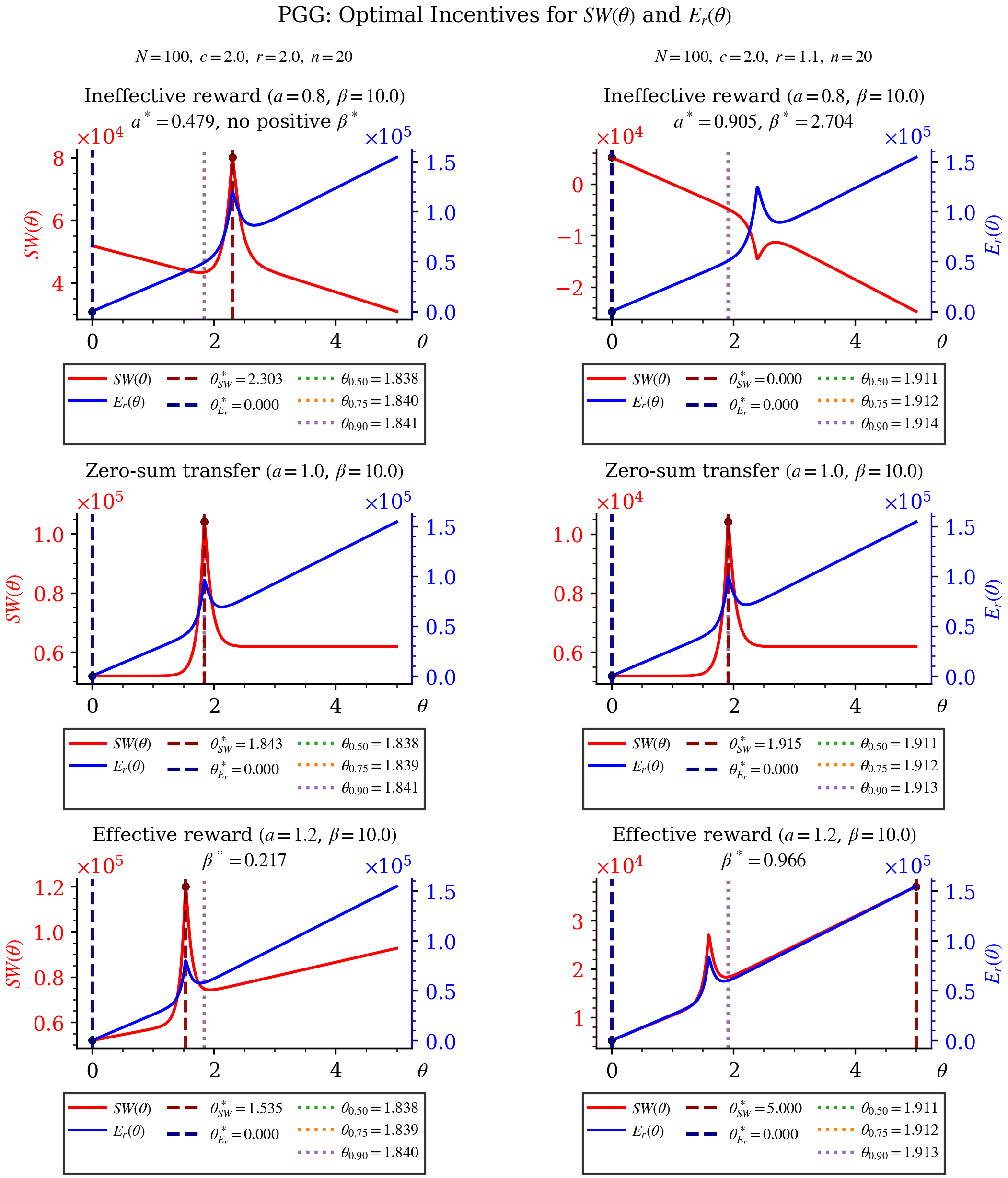}
    \caption{
  \textbf{Efficient rewards can conflict with cost: multi-objective comparison of optimal incentive levels for the PGG (reward case, $\beta=10.0$).} The panels show that the social-welfare–maximising incentive can differ substantially from the cost-minimising incentive. The figures illustrate how the incentive values that maximise $SW(\theta)$, minimise $E_r(\theta)$, and satisfy cooperation-frequency thresholds vary across three reward-efficiency regimes—ineffective reward ($a<1$), zero-sum transfer ($a=1$), and effective reward ($a>1$). Vertical lines mark the optimal incentives and the threshold values for target cooperation frequencies.}
    \label{fig:SW-vs-Er-PGG}
\end{figure}
This section compares the optimal incentive levels to three inter-related optimisation  objectives: social welfare maximisation, institutional cost minimisation, and institutional cost minimisation with a minimum target cooperation frequency. It is crucial to focus on whether the social welfare optimum remains feasible and compatible with the institutional cost, once the frequency of cooperation threshold is imposed. From Figures \ref{fig:SW-vs-Er-DG} and \ref{fig:SW-vs-Er-PGG}, for DG and PGG, respectively, we show that the optimal incentives $\theta^*$ for achieving the objectives can be significantly different, for varying  game parameters (see additional figures A3-A8 in Appendix for other parameter settings).

\textbf{(i) Institutional cost minimisation under the constraints of the frequency of cooperation.} 

The frequency of cooperation imposes a lower bound $\theta_\omega$ on the incentive level (see  Equation \ref{eq:omega_fraction}). Since $E_r(\theta)$ increases with the incentive level in the relevant regimes, institutional cost does not favour incentives above this threshold. Hence, once a target cooperation frequency is fixed, the cost-minimising feasible incentive is naturally located at $\theta_\omega$. This shows that the cost objective and the frequency of cooperation should be interpreted jointly: the threshold defines feasibility, while $E_r(\theta)$ selects the lowest feasible incentive.

\textbf{(ii) Comparison with the social welfare maximisation objective.}

The central comparison is between the social welfare optimum and the feasible incentive to minimise institutional cost. If $\theta^*_{SW}$ lies close to $\theta_\omega$, then social welfare maximisation, cost minimisation, and cooperation enforcement are largely compatible. If $\theta^*_{SW}>\theta_\omega$, then maximising social welfare requires incentives above the minimum level needed to satisfy the cooperation target, creating a trade-off with institutional cost. Conversely, if $\theta^*_{SW}<\theta_\omega$, the frequency of cooperation forces the institution to choose an incentive level beyond the welfare optimum.

In the inefficient and zero-transfer regimes, $a \le 1$, the interaction between the objectives depends on the shape of $SW(\theta)$. When $a \le a^*\text{, }SW(\theta)$ is non-increasing over the feasible interval, or only admits a local minimiser. As a result, social welfare maximisation and institutional cost minimisation are aligned, since both favour the smallest feasible incentive. When $SW(\theta)$ admits an interior local maximiser ($a > a^*$), the welfare optimum may lie above the cooperation threshold, producing a partial trade-off with the institutional cost. However, if the frequency of cooperation threshold exceeds this local maximiser, the feasible solution is instead dictated by the cooperation requirement.

In the efficient-transfer regime, $a>1$, the social welfare objective tends to favour larger incentive levels, whereas institutional cost minimisation favours the smallest feasible incentive. The frequency of cooperation threshold determines the minimum admissible value of $\theta$, but the welfare optimum may lie above this threshold. This creates a stronger conflict between increasing welfare and controlling institutional cost.

The neutral and strong selection limits provide simplified benchmark cases for the multi-objective comparison. Under both limits, $SW(\theta)$ approaches a linear limiting profile, while the thresholds for cooperation either become extremely large or converge to a common value. These limiting behaviours reduce the complexity of the trade-off among the objectives. 



\section{Conclusions and Future Work}
\label{sec: summary}

In this work, we extend and generalise the classical framework of institutional incentives by putting social welfare at the centre of the analysis, rather than focusing only on minimising institutional costs or maximising cooperation levels \cite{HAN2026_social_welfare,duong2023cost,chen2015first,wang2019exploring}. For institutional reward, we derive an explicit expression for population-level social welfare in finite, well-mixed populations and show that both how efficiently rewards are transferred and how strongly selection acts are crucial for whether rewarding cooperation produces a net societal benefit. We identify parameter regimes where social welfare has a single best incentive level (where it first increases and then decreases as incentives grow), and other regimes where it exhibits more complex phase transitions, with multiple local optima as efficiency and selection change. We also show that any welfare-maximising incentive is either zero or very close to a simple, analytically determined target level. 

For institutional punishment, we find qualitatively similar localisation properties and derive explicit conditions under which sufficiently efficient rewards always outperform punishment in terms of social welfare for any given budget. Taken together, these results demonstrate that incentive schemes optimised for cost efficiency or cooperation frequency can be markedly suboptimal or even harmful when social welfare is considered. This challenges much of the existing literature on the evolution of cooperation and institutional incentives, which typically uses cooperation levels or institutional spending cost as the main performance criteria \cite{hilbe2012emergence,sigmund2001reward,sigmund2010social,cimpeanu2019exogenous,chen2015first,wang2021incentive,sasaki2012take,gois2019reward,duong2023cost,garcia2019evolution,wu2017probabilistic,liu2026dynamic}. Moreover, our findings provide a welfare-based analytical framework for evaluating mechanisms of cooperation and clarify when higher cooperation genuinely improves collective outcomes \cite{HAN2026_social_welfare}. As important implications for institutional design and public policy, our findings underscore the importance of welfare-centric optimisation and offer structural guidance on how to calibrate reward and sanction systems. Our framework suggests that centralised controllers and incentive designers should be judged not just on how much cooperation they induce among agents, but on how well they maximise net social welfare under realistic resource constraints.

Several important gaps remain open for future work. Firstly, as our work develops the reward-based and punishment-based social welfare frameworks separately, a natural next step is to  investigate adaptive and hybrid combinations \cite{chen2015first,liu2025evolution,duong2023cost,liu2026dynamic} with  reward and punishment for enhanced social welfare. Secondly, our numerical results suggest clear structural behaviours of social welfare on several key intervals, but a rigorous global characterisation remains open. Addressing this gap would strengthen the proposed framework and provide a more complete understanding of optimal incentive design in finite populations. Thirdly, our current formulation assumes a uniform incentive level across all targeted individuals, whereas a more realistic institutional design would allow incentives to depend on the specific individual, their role, or their local contribution to the population state \cite{perc2015double,cimpeanu2025hidden,cimpeanu2021cost,alalawi2026trust}. The generalisation to a heterogenous incentives could reveal richer optimal policies and better reflect real social systems. Finally, the model currently assumes symmetric mutant emergence from the two absorbing states. Allowing asymmetric mutation probabilities, or more generally state-dependent mutation structures, would make the framework more realistic and may substantially change both the welfare landscape and the resulting optimal intervention strategies. Together, these extensions would help move the present framework from a stylised welfare analysis toward a more general and policy-relevant theory of institutional incentive design in finite populations. 



\section*{Acknowledgements}
 Z.S. and TAH are supported by EPSRC (grant EP/Y00857X/1). MHD was supported by EPSRC grant EP/Y008561/1. 
T.A.H. acknowledges travel support from the HCMUT\text{-}VNUHCM (Adjunct Professorship scheme HCMUT\text{-}VNUHCM).

\section*{Competing interest} Authors declare that they have no conflict of interest.

\bibliographystyle{unsrt}
\bibliography{mybib}

\newpage

\appendix
\setcounter{figure}{0}
\setcounter{equation}{0}

\renewcommand*{\thefigure}{A\arabic{figure}}
\renewcommand*{\theequation}{A\arabic{equation}}

\section{Proof of the main results: reward}
\label{sec: appendix1}
\subsection{Zero Sum Transfer (\texorpdfstring{$a=1$}{a=1})}
\begin{proof}[Proof of the case of zero-sum transfer ($a=1$)]
Under the condition $a=1$, the term $(a-1)\theta$ inside the bracket in the right-hand side of \eqref{eq:SW-theta-final} vanishes. Thus the expected total social welfare, $SW(\theta)$, is simplified to:
\[
SW(\theta)= K  \frac{f(x)}{g(x)},
\]
where we recall that $x=\beta(a\theta+\delta)=\beta(\theta+\delta)$ and $K:= \frac{N^2}{2} (\delta + N\Delta)$. 
For the DG game:
\begin{equation*}
K=\frac{N^2}{2}(\delta+N\Delta)=\frac{N^2}{2}\Bigg[-\Big(c+\frac{b}{N-1}\Big)+\frac{N b}{N-1}\Bigg]=\frac{N^2}{2}(b-c)>0\quad \text{since}\quad b>c.    
\end{equation*}
For the PGG game:
\begin{equation*}
K= \delta+N\Delta=\frac{N^2}{2}\Bigg[-c \left(1 - \frac{r(N-n)}{n(N-1)} \right)+\frac{N rc(n-1)}{n(N-1)}\Bigg]=\frac{N^2}{2} c(r-1)>0\quad\text{since}\quad c>0,~ r>1.
\end{equation*}
Thus $K>0$ for both games.
Therefore, the problem of maximising $SW(\theta)_{a=1}$ reduces to the following optimisation problem
\[
\max_{\theta\geq \theta_\omega}\Psi(\theta),\quad\text{where}\quad\Psi(\theta):=\frac{f(x)}{g(x)}=\frac{f(\beta(\theta+\delta))}{g(\beta(\theta+\delta))}. 
\]
Using the chain rule, we compute the derivative of $\Psi$:
\begin{equation}
\label{eq: derivative Psi}
\Psi'(\theta)=x'(\theta)\frac{d}{dx}\Big(\frac{f(x)}{g(x)}\Big)=\beta\frac{f'(x)g(x)-f(x)g'(x)}{g(x)^2}:=\beta\frac{u P(u)}{g^2(x)},    
\end{equation}
where following \cite{duonghan2021cost} we have defined $u= e^x$ and
\begin{equation}
P(u) = \dfrac{f(x)g'(x)-f'(x)g(x)}{u}.    
\end{equation}
More precisely,
\begin{align}
\label{eq: def of P}
P(u)&:=(1+u)\Bigg[\Big(\sum_{j=0}^{N-2}(H_N+\frac{1}{N-1-j}) u^j\Big)\Big(\sum_{j=1}^{N-1} j u^{j-1}\Big)-\Big(\sum_{j=1}^{N-2}\Big(H_N + \frac{1}{N-1-j}\Big) j u^{j-1}\Big)\Big(\sum_{j=0}^{N-1}u^j\Big)\Bigg]\notag
\\&\qquad-\Big(\sum_{j=0}^{N-2}\Big(H_N + \frac{1}{N-1-j}\Big) u^j\Big)\Big(\sum_{j=0}^{N-1}u^j\Big),
\end{align}
where $H_N$ denotes the harmonic number
\[
H_N=\sum_{j=1}^{N-1}\frac{1}{j}.
\]
We will use the following properties of $P$ which is proved in  \cite[Proposition 1.6, supplementary document]{duonghan2021cost}
\begin{itemize}
    \item[(P1)] $P(u)$ is a polynomial of order $2N-4$.
    \item[(P2)] For $u > 0$, $P(u)$ has exactly one solution $u_0 > 1$ and $\textrm{sign}(P(u)) = \textrm{sign}(u - u_0)$.
\end{itemize}
From \eqref{eq: derivative Psi} and the second property above, we have $\Psi'(\theta)$ has a unique positive root $\theta_0 =\frac{\log(u_0)}{\beta}-\delta$. In addition, $\Psi(\theta)$ is increasing on $(0, \theta_0)$ and decreasing on $(\theta_0, +\infty)$. This guarantees that $\Psi$ has a global maximum at $\theta_0$. 
\end{proof}
\subsection*{Numerical calculation of the maximiser}
We have established the existence of a unique global maximiser for the expected total welfare $SW(\theta)$ in the case $a=1$. The maximiser $\theta_0$ which is computed from the unique positive root $u_0>1$ of the polynomial $P$ defined in \eqref{eq: def of P}. However, for large $N$, according to Abel's impossibility theorem, finding the analytical value for $u_0$ (thus $\theta_0$) is analytically intractable since $P(u)$ is a polynomial of degree $2N-4$.

Therefore, we compute $u_0$ numerically using Brent’s method, which combines bisection, secant, and inverse quadratic interpolation to ensure both reliability and fast convergence \cite{numericalrecipes}. Since $P(u)$ is a polynomial admitting a unique positive root, Brent’s method can be used to compute $u_0$ accurately and efficiently once a valid bracketing interval is established. To this end, we first determine an interval $[u_{\min}, u_{\max}]$ such that $P(u_{\min})P(u_{\max}) < 0$. Noting that $P(1) < 0$, we construct the upper bound by iteratively doubling:
\[
u_{\max} = 2^p,\quad \text{where } p = \min \{k \in \mathbb{Z}^+ : P(2^k) > 0\},
\]
thereby ensuring the existence of a sign change within the interval.

The maximum overall social welfare is achieved when the incentive $\theta$ forces the system into the state defined by $x_0$. Since $x = \beta(a\theta + \delta)$, the Optimal Social Welfare Incentive ($\theta_0^{SW(\theta)}$) is derived as:\[\theta_0^{SW(\theta)}=\dfrac{x_0}{a\beta}-\frac{\delta}{a}\]

\subsection{For $a>1$}

\begin{proof}[Proof of the case of sufficient transfer ($a>1$)]

To determine the optimal incentive $\theta$ that maximises social welfare, we analyse the derivative of $SW(\theta)$ with respect to $\theta$. Substituting the expressions for $f(x)$, $g(x)$, and their derivatives, the derivative $dSW(\theta)/d\theta$ can be computed explicitly as follows:
We take the derivative of Social Welfare objective function from  \ref{eq:SW-theta-final}, for $\theta$ such that $u > u_0$, as follows:
\begin{align}
    \frac{dSW(\theta)}{d\theta} &= \frac{N^2}{2}\left[ a\beta\frac{f'(x)g(x) - f(x)g'(x)}{g^2(x)}\left(\delta + N\Delta + (a-1)\theta\right) + (a-1)\frac{f(x)}{g(x)}\right]\notag \\
    &= \frac{N^2(a-1)}{2g^2(x)}\left[a\beta(f'(x)g(x)-f(x)g'(x)) \left( \frac{\delta + N\Delta}{a-1} + \theta \right) + f(x)g(x) \right]\notag\\
&= \frac{N^2(a-1)uP(u)}{2g^2(x)}\left[-a\beta \left( \frac{\delta + N\Delta}{a-1} + \theta \right) + \frac{f(x)g(x)}{uP(u)} \right]\notag\\
    &= \frac{N^2(a-1)uP(u)}{2g^2(x)} \left[\frac{f(x)g(x)}{uP(u)} - a\beta\theta + a\beta\frac{\delta + N\Delta}{1-a} \right] \notag
\\&= \frac{N^2 u P(u)}{2 g^2(x)} (a-1) \left[ F(u) + \beta \mathcal{K} \right],\label{dSW}
\end{align}
where we have defined
\begin{equation}
\label{eq: def F}
F(u):=\frac{f(x)g(x)}{uP(u)} - x\quad
    \text{and } \mathcal{K}:=\frac{\delta + aN\Delta}{1-a}.
\end{equation} 
From \cite{duonghan2021cost} we have important properties of $F(u)$ on $(u_0, +\infty)$, where $u_0$ is the unique positive root of the polynomial $P$,  as follows:
\begin{enumerate}
    \item[(F1)] $F(u)>0$,
    \item[(F2)] $F(u) > 1+u-\log(u)$. In particular, $ \lim{F(u)} = +\infty$.
    \item[(F3)] There $u^*$ such that $ F^*:=F(u^*)=\min_{u>u_0}F(u)$.
    \item[(F4)] When $N \le N_0 = 100$, $u^*$ is unique and $\textrm{sign}(dF/du)=\textrm{sign}(u-u^*)$.
\end{enumerate}

We consider the sign of $\mathcal{K}$. For Donation Game:
\begin{align*}
    \mathcal{K} &= \frac{1}{1-a}\left( -c-\frac{b}{N-1}+\frac{aNb}{N-1}\right) \\
    &= \frac{Nb}{(1-a)(N-1)}\left[ a -\frac{c(N-1)}{Nb} - \frac{1}{N} \right] \\
    &= \frac{Nb}{(1-a)(N-1)}\left[ a -\frac{c(N-1) + b}{Nb} \right].
\end{align*}
For Public Good Game:
\begin{align*}
    \mathcal{K} &= \frac{1}{1-a}\left[ -c +c\frac{r(N-n)}{n(N-1)}+aN\frac{rc(n-1)}{n(N-1)}
    \right] \\
    &= \frac{c}{(1-a)n(N-1)} \left[  -n(N-1) + r(N-n)+aNr(n-1)\right]\\
    &= \frac{cNr(n-1)}{(1-a)n(N-1)} \left[ a - \frac{n(N-1)-r(N-n)}{Nr(n-1)} \right].
\end{align*}
We define the threshold value of $a$
\begin{equation*}
a^*:=\begin{cases}
    \frac{c(N-1) + b}{Nb}\quad \text{in Donation Game},\\
    \frac{n(N-1)-r(N-n)}{Nr(n-1)}\quad \text{in Public Goods Game}.
\end{cases}
\end{equation*}
Then it follows that in both games $\textrm{sign}(\mathcal{K})=\textrm{sign}(1-a)\textrm{sign}(a-a^*)$. Note further that since \begin{equation*}
    \begin{cases}
        c(N-1)+b >0 \\
        [c(N-1)+b] - Nb = (c-b)(N-1)<0\\
        n(N-1) - r(N-n) = (n-r)(N-1)+r(n-1) > 0\\
        [n(N-1) - r(N-n)] - Nr(n-1) = n(N-1)(1-r) < 0
    \end{cases}
\end{equation*}
we deduce that $0 < a^* < 1$.

In the case of sufficient transfer, since $a > 1 > a^*$, $\mathcal{K}$ is always negative. We recall the threshold value $\beta^*$
\begin{equation*}
            \beta^* = -\frac{F^*}{\mathcal{K}}>0.
        \end{equation*}
(i) Then for $\beta \le \beta^*$, $SW(\theta)$ is non-decreasing on $(\theta_0, +\infty)$, where we recall that $\theta_0 = \frac{\log u_0 - \beta \delta}{\beta a}$.

(ii) For $\beta > \beta^*$, the number of changes of the sign of $dSW(\theta)/d\theta$ is at least two for all $N$ and there exists an $N_0$ such that the number of changes is exactly two for $N \le N_0$. As a consequence, for $N \le N_0$, there exist $\theta_1 < \theta_2$ such that, for $\beta > \beta^*$, $SW(\theta)$ is increasing when $\theta < \theta_1$, decreasing when $\theta_1 < \theta < \theta_2$ and increasing when $\theta > \theta_2$. 
\end{proof}
\subsection{For $a<1$}
\begin{proof}[Proof of the case of insufficient transfer ($a<1$)]

(i) In the case of insufficient transfer $a<1$, then by the definition of $a^*$, we have $\mathcal{K} > 0$ when $a > a^*$ and $\mathcal{K} < 0$ when $a < a^*$. For $\mathcal{K} > 0$, $SW(\theta)$ is strictly decreasing on $(\theta_0, +\infty)$.

(ii) With the threshold $\beta^*$ from \ref{beta_star}, for $\beta < \beta^*$, $SW(\theta)$ is non-decreasing on $(\theta_0, +\infty)$. Consequently:
        \[
        \max_{\theta \ge \theta_0} SW(\theta)=SW(\theta_0)
        \]

(iii) For $\beta > \beta^*$, the number of changes of the sign of $dSW(\theta)/d\theta$ is at least two for all $N$ and there exists an $N_0$ such that the number of changes is exactly two for $N \le N_0$. As a consequence, for $N \le N_0$, there exist $\theta_1 < \theta_2$ such that, for $\beta > \beta^*$, $SW(\theta)$ is decreasing when $\theta < \theta_1$, increasing when $\theta_1 < \theta < \theta_2$ and decreasing when $\theta > \theta_2$. Thus, for $N \le N_0$:
        \[
        \max_{\theta \ge \theta_0}SW(\theta)=\max\{SW(\theta_0),SW(\theta_2)\}  
        \]

(iv) Moreover, for sufficiently large $\beta$ and small $\theta$, $SW(\theta)$ is increasing as $\theta \to 0^+$ (shown in Lemma \ref{lem:lower_bound_c}).

\end{proof}
\subsection*{Proof of Theorem \ref{thm:localization}} \label{proof-thm}
In this section, we will prove Theorem \ref{thm:localization}. To this end, we first need some axillary results.  The following proposition presents some properties of the expected total social welfare $SW$.
\begin{proposition}[Basic properties of $SW(\theta)$]
\label{prop:sw_basic}
Recall the Social Welfare objective
\[
SW(\theta)= \frac{N^2}{2}\frac{f(x)}{g(x)} \big(\delta + N\Delta + (a-1)\theta\big) 
\]
where $x = \beta(a\theta+\delta)$. With $a < 1$, the following properties hold:
\begin{enumerate}
    \item \textbf{Positive interval:} $
    SW(\theta) \ge 0 \text{ ~ iff ~ } \theta \in I := \left[ 0, \frac{\delta+N\Delta}{1-a}\right]$.
    
    \item \textbf{Local extrema:} $SW'(\theta)=0$ if and only if:
    \begin{align}
         \frac{1-a}{\delta+N\Delta+(a-1)\theta}
    =-a\beta\frac{uP(u)}{f(x)g(x)}. \label{critical points}
    \end{align}   
\end{enumerate}
\end{proposition}
\begin{proof}[Proof of Proposition \ref{prop:sw_basic}]
The first statement follows directly from the formula of $SW(\theta)$. In fact, since $f(x)$ and $g(x)$ are positive polynomials, the sign of $SW(\theta)$ is the same as the sign of factor $\delta + N\Delta + (a-1)\theta$.
For the second statement, we recall from \eqref{dSW}
that 
\[
    \frac{dSW(\theta)}{d\theta} = \frac{N^2(a-1)}{2g^2(x)} \left[ a\beta uP(u) \left( \frac{\delta + N\Delta}{1-a} - \theta \right) + f(x)g(x) \right],
    \]
    where $u=e^x$. From this we deduce \eqref{critical points}.
\end{proof}
In the following proposition, we give an alternative representation for $f(x)$ and $g(x)$.
\begin{proposition}
\label{prop:f and g}
The functions $f(x)$ and $g(x)$ can be expressed in the following forms:
\begin{equation*}
    f(x) = \sum_{j=0}^{N-1} \eta_j u^j, \quad\text{and}\quad g(x) = \sum_{j = 0}^{N-1} u^j,
\end{equation*}    
where $u(\theta) = e^{\beta(a\theta + \delta)} = e^{x}$ and
\[
\eta_0 = \frac{1}{N-1} + H_{N},\quad \eta_j = 2H_{N} + \frac{1}{N-j} + \frac{1}{N-j-1} \quad \text{for } 1 \le j \le N-2,\quad \text{and}\quad
    \eta_{N-1} = 1 + H_{N}.
\]    
\end{proposition}
\begin{proof}[Proof of Proposition \ref{prop:f and g}]
This follows directly from the formula of $f(x)$ and $g(x)$ in \eqref{fx}  and \eqref{gx} respectively. 
\end{proof}
Let $S(u) = \frac{f(x)}{g(x)}$ and for convenience, $R(\theta) = S(u(\theta))$. Specifically, let:
\[
S(u) = \frac{\sum_{j=0}^{N-1} \eta_j u^j}{\sum_{j=0}^{N-1} u^j} \qquad \text{ and } \qquad R(\theta) = \frac{\sum_{j=0}^{N-1} \eta_j e^{j\beta(a\theta+\delta)}}{\sum_{j=0}^{N-1} e^{j\beta(a\theta+\delta)}}.
\]
\begin{lemma} \label{obs_eta0}
$S(u) > \eta_0$ for all $u > 0$.
\end{lemma}
\begin{proof}[Proof of Lemma \ref{obs_eta0}]
It follows from the formula of $\eta_j$ that $\eta_j > \eta_0$ for all $1 \le j \le N-1$. Therefore, we have $f(u) > \sum_{j=0}^{N-1} \eta_0 u^j = \eta_0 g(u)$, which implies $S(u) > \eta_0$.
\end{proof}
\begin{lemma} \label{obs_lipschitz}
The function $S(u)$ is Lipschitz continuous and strictly increasing for all $u \in (0, 1)$.
\end{lemma}
\begin{proof}[Proof of Lemma \ref{obs_lipschitz}]
By the chain rule, we have:
\begin{align*}
\frac{dS}{du} &= \frac{dS}{dx} \frac{dx}{du} = \frac{-uP(u)}{g(x)^2} \frac{1}{u} = \frac{-P(u)}{g(x)^2}
\end{align*}

Thus, $dS/du$ is continuous on $(0, 1)$. Furthermore, the sign of $dS/du$ matches the sign of $-P(u)$. Since $P(u) < 0$ for all $u \in (0, u_0)$ and $u_0 > 1$, $S(u)$ is strictly increasing for all $u \in (0, 1)$.
\end{proof}

Let $L$ be Lipschitz bound for $S(u)$ in the interval $[0, 1]$.  The property stated in the following Lemma yields the lower bound for the search interval of Algorithm \ref{alg:pgg_main}

\begin{lemma}[Left-side dominance by $SW(0)$]
\label{lem:lower_bound_c}

There exists a constant $r_l > 1$, independent of $\beta$, such that, for all sufficiently large $\beta$,  $SW(\theta) \le SW(0)$ for all $\theta \in [0, \mu]$ where $\mu = \frac{-\delta}{a} - \frac{\log r_l}{a\beta}$.
\end{lemma}
To prove lemma \ref{lem:lower_bound_c}, we divide $[0, \mu]$ into two parts and use the following two claims.
\begin{claim} \label{claim:decrease_near_zero}
For sufficiently large $\beta$, $SW(\theta)$ is strictly decreasing on the interval $I_0 := \left[0, \frac{1}{a\beta}\right]$.
\end{claim}
\begin{proof}[Proof of Claim \ref{claim:decrease_near_zero}]
First, we make sure that $\beta$ is large enough such that $u(\frac{1}{a\beta}) = e^{\beta\delta+1} < 1$, so that $S'(u(\theta)) > 0$, and so is $R'(\theta)=u'(\theta)S'(u(\theta))$, for all $\theta \in I_0$.

We proceed by contradiction. Suppose there exists a point $\tilde{\theta} \in I_0$ such that $SW'(\tilde{\theta}) \ge 0$. Differentiating the objective function $SW(\theta)$ yields:
\[ 
SW'(\theta) = \frac{N^2}{2}\left( (\delta + N\Delta)R'(\theta) - (1-a)\theta R'(\theta) - (1-a)R(\theta) \right).
\]
For $SW'(\tilde{\theta}) \ge 0$ to hold, we require:
\[
(\delta + N\Delta)R'(\tilde{\theta}) - (1-a)\tilde{\theta} R'(\tilde{\theta}) \ge (1-a)R(\tilde{\theta})
\]

We know $a < 1$, and $u'(\theta) = a\beta e^{\beta(a\theta + \delta)} \le a\beta e^{\beta \delta + 1}$ for all $\theta \in I_0$. Combining these facts with Lemma \ref{obs_eta0} and Lemma \ref{obs_lipschitz}, the inequality implies:
\[
(\delta + N\Delta)a\beta e^{\beta\delta + 1} L \geq (\delta + N\Delta)R'(\tilde{\theta}) - (1-a)\tilde{\theta} R'(\tilde{\theta}) \ge (1-a) R(\tilde{\theta}) > (1-a)\eta_0.
\]
Since $\delta < 0$, we have  $\lim_{\beta \to \infty} \beta e^{\beta\delta} = 0$. Thus, for sufficiently large $\beta$, the above inequality fails. Thus, we conclude that for sufficiently large $\beta$, it holds that $SW'(\theta) < 0$ for all $\theta \in I_0$. This implies $SW(\theta) < SW(0)$ for all $\theta \in I_0$ for sufficiently large $\beta$.
\end{proof}

\begin{claim} \label{claim:bound_away_from_zero}
There exists an $\varepsilon_0 > 1$, independent of $\beta$, such that for all $\theta \in I_1:=\left[\frac{1}{a\beta}, \frac{-\delta}{a} - \frac{\log \varepsilon_0}{a\beta}\right]$, $SW(\theta) \le SW(0)$.
\end{claim}
\begin{proof}[Proof of Claim \ref{claim:bound_away_from_zero}]
First, note that this interval is valid for sufficiently large $\beta$ since $\frac{1}{a\beta} \le \frac{-\delta}{a} - \frac{\log \varepsilon_0}{a\beta}$, for $\beta \ge \frac{1 + \log \varepsilon_0}{-\delta}$.

Suppose by contradiction that there exists a $\tilde{\theta} \in I_1$ such that $SW(\tilde{\theta}) > SW(0)$. This is equivalent to:
\[
(\delta + N\Delta)R(\tilde{\theta}) - (1-a) \tilde{\theta} R(\tilde{\theta}) > (\delta + N\Delta) R(0).
\]
Rearranging the terms yields:
\begin{align}
(\delta + N\Delta)(R(\tilde{\theta}) - R(0)) > (1-a)\tilde{\theta} R(\tilde{\theta}),\label{ineq:lower_bound} 
\end{align}

which implies $R(\tilde{\theta}) - R(0) > 0$. Applying Lemma \ref{obs_eta0} and Lemma \ref{obs_lipschitz} yields:
\begin{align}
    (1-a)\tilde{\theta} \eta_0 < (1-a)\tilde{\theta} R(\tilde{\theta}) &< (\delta + N\Delta)(R(\tilde{\theta}) - R(0)) \notag \\
    &< (\delta + N\Delta)L(u(\tilde{\theta}) - u(0)) < (\delta + N\Delta)Lu(\tilde{\theta}).\label{ineq:3sides}
\end{align}
By making the substitution $\tilde{\theta} = \frac{-\delta}{a} - \frac{\log \varepsilon}{a\beta}$ (where $\varepsilon_0 \le \varepsilon \le e^{-\beta\delta - 1}$), we have $ u(\tilde{\theta}) = \varepsilon^{-1}$. Rearranging the left-most and right-most sides of the inequality, we have:
\[
\varepsilon \left(-\delta- \frac{\log \varepsilon}{\beta}\right) < \frac{(\delta + N\Delta)La}{(1-a)\eta_0}.
\]
Observe that the left-hand side is an increasing function of $\varepsilon$, for all $1 \le \varepsilon \le e^{-\beta\delta - 1}$. Therefore, by choosing an $\varepsilon_0$ strictly independent of $\beta$ such that $-\delta \varepsilon_0 \ge \frac{(\delta + N\Delta)La}{(1-a)\eta_0}$, and a $\beta$ large enough such that $\varepsilon_0$ is valid (i.e. $1 \le \varepsilon_0 \le e^{-\beta\delta-1}$) the inequality fails for all valid $\varepsilon \ge \varepsilon_0$, a contradiction. Hence, proving our claim.
\end{proof}
\begin{proof}[Proof of Lemma \ref{lem:lower_bound_c}]
Lemma \ref{lem:lower_bound_c} follows by combining Claim \ref{claim:decrease_near_zero} and Claim \ref{claim:bound_away_from_zero} and
setting $r_l = \varepsilon_0$.
\end{proof}


We are now in the position to prove Theorem \ref{thm:localization}.
\begin{proof}[Proof of Theorem \ref{thm:localization}]

We establish the upper and lower bound of the search interval.
\textit{For the upper bound:}  Since the polynomial $P(u)$ has a unique root $u_0 > 1$, from \eqref{critical points}, for every $\theta \in I$, the LHS is strictly positive and strictly increasing with respect to $\theta$, while the RHS is strictly negative for all $u > u_0$. Therefore, any extrema $\theta^* \in I$ (i.e. solution to equation \eqref{critical points}) must yield $u \in (0, u_0)$. This implies that $ e^{\beta(a\theta^* + \delta)} < u_0$
which is equivalent to
\begin{equation*}
\theta^* - \left( \frac{-\delta}{a} \right) = \theta^* - \theta_{\infty} < \frac{\log u_0}{a\beta}.  
\end{equation*}
This gives the upper bound.

The lower bound is a consequence of Lemma \ref{lem:lower_bound_c}. In fact, since on $[0, \mu]$, $SW(\theta)\leq SW(0)$), any meaningful extrema $\theta^*$ must satisfy $\theta^* > \mu$, or equivalently
\[
\theta^* - \frac{-\delta}{a} > -\frac{\log r_l}{a\beta}.
\]
This establishes the lower bound of $\theta^*$ and completes the proof of Theorem \ref{thm:localization}.
\end{proof}
\subsection{Influence of Selection Intensity (\texorpdfstring{$ \beta $}))} \label{beta_theory}
The intensity of selection $\beta$ also plays a crucial role in deciding the overall stability and long\text{-}term optimal structure of the system's Social Welfare ($SW$).

We recall from \eqref{eq:SW-theta-final} that
\[
SW(\theta)=\frac{N^2}{2}\frac{f(x)}{g(x)}(\delta+N\Delta+(a-1)\theta),
\]
where $x=\beta(a\theta+\delta)$.
\subsection*{Neutral Selection Limit ($\beta \to 0^+$)}
The weak selection (neutral) limit corresponds to $\beta\to 0$, where payoff differences have only a small effect on strategy adoption probabilities. In this regime, updates are nearly random and the dynamics are dominated by neutral drift, with selection acting only as a weak perturbation. 

It follows from the above formula for the total expected social welfare and Proposition \ref{prop:f and g} that
\[ 
\lim_{\beta\to 0} SW(\theta)=\frac{N^2}{2}\frac{f(0)}{g(0)}(\delta+N\Delta+(a-1)\theta)=\frac{N^2}{2}\frac{\sum_{j=0}^{N-1}\eta_j}{N}(\delta+N\Delta+(a-1)\theta)=N^2 H_N (\delta+N\Delta+(a-1)\theta).
\]

\subsection*{Strong Selection Limit ($\beta \to +\infty$)}
By a straightforward adaption of the proof of \cite[Proposition 1.12]{duonghan2021cost} we have
\begin{equation*}
    \lim_{\beta \to +\infty}\frac{f(x)}{g(x)}=\begin{cases}
      2 H_N\quad\text{for}\quad \theta=-\frac{\delta}{a},\\
      H_N+1\quad\text{for}\quad \theta>-\frac{\delta}{a},\\
      H_N+\frac{1}{N-1}\quad\text{for}\quad \theta<-\frac{\delta}{a}.
    \end{cases}    
    \end{equation*}
It immediately follows that
\begin{equation*}
    \lim_{\beta \to +\infty}SW(\theta)=\begin{cases}
      N^2 H_N (\delta+N\Delta+(a-1)\theta)\quad\text{for}\quad \theta=-\frac{\delta}{a},\\
      \frac{N^2}{2}(H_N+1)(\delta+N\Delta+(a-1)\theta)\quad\text{for}\quad \theta>-\frac{\delta}{a},\\
      \frac{N^2}{2}\Big(H_N+\frac{1}{N-1}\Big) (\delta+N\Delta+(a-1)\theta)\quad\text{for}\quad \theta<-\frac{\delta}{a}.
    \end{cases}    
    \end{equation*}

\section{Proof of the  main results: punishment}
\label{sec: appendix2}
\subsection{Social welfare in institutional punishment}
\label{sec: punishment calculations}
In our model where the institution punishes Defectors (instead of rewarding Cooperators in above sections), the total payoff received in the population with $i$ cooperators and incentive efficiency $\hat{a}$, is: 
\[
\hat{P}_i = i\,\Pi_C(i) + (N - i)\big[ \Pi_D(i) - \hat{a}\theta \big].
\]
Combining with the total institutional cost $\hat{\theta}_i = (N - i)\theta$ yields aggregate the social welfare in state $S_i$ as:
\begin{equation*}
\widehat{SW}_i(\theta) = \hat{P}_i - \hat{\theta}_i =  i\,\Pi_C(i) + (N - i)\big[ \Pi_D(i) - (1+\hat{a})\theta \big]
\end{equation*}
Which then by carrying out the same algebraic process as the reward case, yields the Social Welfare function:
\begin{align}
\nonumber
\widehat{SW}(\theta) &=\frac{1}{2}\sum_i \widehat{SW}_i(\theta)(n_{1,i} + n_{N-1,i}) \\
&= \frac{1}{2}\sum_{i}( i(\delta + N\Delta) - (N - i)(1+\hat{a})\theta ) 
(n_{1,i} + n_{N-1,i}) \nonumber \\ 
\nonumber
&= \frac{1}{2}(\delta + N\Delta) \sum_{i}i (n_{1,i} + n_{N-1,i}) - \frac{1}{2}(1+\hat{a})\theta\sum_{i}(N-i) (n_{1,i} + n_{N-1,i})
\\
&= \frac{N^2}{2}\frac{f(x)}{g(x)} \big(\delta + N\Delta) - \frac{N^2}{2}\frac{\hat{f}(x)}{g(x)}(1+\hat{a})\theta, \label{eq:swp}
\end{align}
where as derived in \cite{duonghan2021cost}
\[
\hat{f}(x) = (1 + e^x)\left[ \left(1 + e^x + \cdots + e^{(N-2)x}\right) H_N 
+ \sum_{j=1}^{N-1} \frac{e^{(j-1)x}}{j} \right]
\]
Moreover, one can write $\hat{f}(x)$ in the following form:
\[
\hat{f}(x) = \sum_{j=0}^{N-1} \hat{\eta}_j u^j,
\]
where $u = e^{\beta(\hat{a}\theta + \delta)} = e^x$, and $\hat{\eta}_j = \eta_{N-1-j}$ for all $0 \le j \le N-1$.

\subsection{Algorithm \ref{alg:pgg_main} for punishment}
We claim that Algorithm \ref{alg:pgg_main} also works for the punishment case. We will justify the our algorithm via the following theorem, which is the counterpart of Theorem \ref{thm:localization} in the reward case.

The proof will be carried out in a similar manner to the reward case. That is, by providing an upper and lower bound for the quantity $|\theta^* - \hat{\theta}_{\infty}|$.

Firstly, by applying the chain rule, we have:
\[
\frac{d\widehat{SW}}{d\theta} = \frac{d\widehat{SW}}{dx} \frac{dx}{d\theta}
\]

Since $x = \beta(\hat{a}\theta + \delta)$, we have $\frac{dx}{d\theta} = \hat{a}\beta$. On the other hand, from equation \ref{eq:swp}, we have:
\begin{align*}
    \frac{d\widehat{SW}}{dx} &= \frac{N^2}{2} (\delta + N\Delta) \frac{f'(x)g(x) - f(x)g'(x)}{g(x)^2} - \frac{N^2}{2} (1+\hat{a}) \left( \frac{\hat{f}'(x)g(x) - \hat{f}(x)g'(x)}{g(x)^2} \theta + \frac{\hat{f}(x)}{g(x)} \frac{1}{\hat{a}\beta} \right) \\
    &= \frac{N^2(1+\hat{a})}{2g(x)^2} \left( - \frac{\delta + N\Delta}{1+b}uP(u) + \theta u \hat{P}(u) - \frac{\hat{Q}(u)}{\hat{a}\beta} \right)
\end{align*}

where $u\hat{P}(u) = \hat{f}(x)g'(x) - \hat{f}'(x)g(x)$ and $\hat{Q}(u) = \hat{f}(x)g(x)$ are polynomials defined in \cite{duonghan2021cost}. Therefore, we have:

\begin{equation}
\frac{d\widehat{SW}}{d\theta} = \frac{N^2 \hat{a} \beta (1+\hat{a})}{2g(x)^2} \left( \theta u \hat{P}(u) - \frac{\delta + N\Delta}{1+\hat{a}}uP(u)  - \frac{\hat{Q}(u)}{\hat{a}\beta} \right) \label{eq:swp_derivative}
\end{equation}

\textbf{Upper Bound} \\
We will prove the following lemma:
\begin{lemma} \label{lem:}
There exists a positive number $p_u$, independent of $\beta$, such that for large enough $\beta$, if $\theta$ is an extrema of $\widehat{SW}$ then
\[
\theta < \frac{-\delta}{\hat{a}} + \frac{\log p_u \beta}{\hat{a}\beta}
\] \label{lem:upper_bound_punish}
\end{lemma}
\begin{proof}
Two results from \cite{duonghan2021cost} showed that:
\[
\hat{P}(u)= -u^{2N-4} P \left( \frac{1}{u} \right)
\]
and
\[
\hat{Q}(u) \ge \frac{4}{3(N-1)}(1+u)u\hat{P}(u) \quad \forall u >0
\]

And since $P(u)$ has a unique positive root $u_0 > 1$, $\hat{P}(u)$ also has a unique positive root $u^{*}_0 = \frac{1}{u_0} < 1$. Furthermore, since $P(u)$ is strictly negative for all $u \in (0, u_0)$ and strictly positive for all $u \in (u_0, +\infty)$, it follows that $\hat{P}(u)$ is strictly negative for all $u \in (0, \frac{1}{u_0})$ and strictly positive for all $u \in (\frac{1}{u_0}, +\infty)$.

Consider $\theta$ such that $u(\theta) \in (u_0, +\infty)$ where we define $u(\theta) = e^{\beta(\hat{a}\theta + \delta)}$, a bit different from the reward case. From (\ref{eq:swp_derivative}), it holds that if $\tilde{\theta}$ (corresponding to $\tilde{u}$) is an extrema of $\widehat{SW}(\theta)$ then:
\begin{equation}
\tilde{\theta} \tilde{u} \hat{P}(\tilde{u}) = \frac{\delta + N\Delta}{1+\hat{a}}\tilde{u}P(\tilde{u}) + \frac{\hat{Q}(\tilde{u})}{\hat{a}\beta} > \frac{\hat{Q}(\tilde{u})}{\hat{a}\beta} \ge \frac{4}{3\hat{a}\beta(N-1)}(1+\tilde{u})\tilde{u}\hat{P}(\tilde{u}) \label{ineq:extrema_punishment}
\end{equation}

Comparing left-most and right-most side of the above inequality yields:
\[
\tilde{\theta} > \frac{4}{3\hat{a}\beta(N-1)}(1+\tilde{u}) > \frac{4\tilde{u}}{3\hat{a}\beta(N-1)}
\]

Let $\tilde{u} = \hat{a}\beta v$, by substituting $\tilde{\theta} = \frac{-\delta}{\hat{a}} + \frac{\log \tilde{u}}{\hat{a}\beta}$ into the above inequality, we have:
\[
\frac{-\delta}{\hat{a}} + \frac{\log \hat{a}\beta}{\hat{a}\beta} + \frac{\log v}{\hat{a}\beta} > \frac{4 v}{3(N-1)}
\]

Using the inequality $x \ge e \log x$ for all $x > 0$ gives us:
\[
\frac{-\delta}{\hat{a}} + \frac{1}{e} + \frac{\log v}{\hat{a}\beta} \ge \frac{-\delta}{\hat{a}} + \frac{\log \hat{a}\beta}{\hat{a}\beta} + \frac{\log v}{\hat{a}\beta} > \frac{4 v}{3(N-1)}
\]

Observe that there exists a constant $v^*$ independent of $\beta$ such that for all large enough $\beta$, the above inequality fails for all $v \ge v^*$, which means that if $\tilde{\theta}$ is an extrema of $\widehat{SW}$, and therefore a root of the expression in (\ref{eq:swp_derivative}), it holds that $u(\tilde{\theta}) < \hat{a}\beta v^*$, and therefore 
\[
\tilde{\theta} < \frac{-\delta}{\hat{a}} + \frac{\log \hat{a}\beta v^*}{\hat{a}\beta}
\]
By setting $p_u = \hat{a}v^*$, we have proven our lemma.
\end{proof}

\textbf{Lower Bound} \\
Let $\hat{S}(u) = \frac{\hat{f}(x)}{g(x)}$ and for convenience, $\hat{R}(\theta) = \hat{S}(u(\theta))$. Specifically, let:
\[
\hat{S}(u) = \frac{\sum_{j=0}^{N-1} \hat{\eta_j} u^j}{\sum_{j=0}^{N-1} u^j} \qquad \text{ and } \qquad \hat{R}(\theta) = \frac{\sum_{j=0}^{N-1} \hat{\eta}_j e^{j\beta(\hat{a}\theta+\delta)}}{\sum_{j=0}^{N-1} e^{j\beta(\hat{a}\theta+\delta)}}
\]

First, notice that Lemma \ref{obs_eta0} also applies for $\hat{S}(u)$, that is $\hat{S}(u) > \eta_0$ for all $u > 0$. Furthermore, we have:
\[
\frac{d\hat{S}}{du} = \frac{d\hat{S}}{dx} \frac{dx}{du} = \frac{-u\hat{P}(u)}{g(x)^2} \frac{1}{u} = \frac{-\hat{P}(u)}{g(x)^2}
\]
which, from the analyses in the Upper Bound section, is positive for all $u \in (0, \frac{1}{u_0})$ and negative for all $u \in (\frac{1}{u_0}, +\infty)$. 

With that, we present a similar Lemma as in the reward case:
\begin{lemma}[Left-side dominance by $\widehat{SW}(0)$]
\label{lem:lower_bound_p}

There exists a constant $p_l > 1$, independent of $\beta$, such that for all sufficiently large $\beta$, the following holds: If $\mu = \frac{-\delta}{\hat{a}} - \frac{\log p_l}{\hat{a}\beta}$, then $\widehat{SW}(\theta) \le \widehat{SW}(0)$ for all $\theta \in [0, \mu]$.
\end{lemma}
\begin{proof}
The proof is carried out similarly like in the reward case, with the only difference being that the quantity $(1-a)R(\theta)$ is replaced with $(1+\hat{a})\hat{R}(\theta)$.

\end{proof}
\subsection{Comparison between reward and punishment}
In this section, we demonstrate that rewarding cooperators frequently yields better outcomes than punishing defectors. Specifically, we observed that if the transfer efficiency of rewards is sufficiently high relative to punishments, any budget allocated to punishment can be replaced by a corresponding reward level that achieves equal or greater Social Welfare.

We define the following functions for all $\theta \ge 0$ and incentive efficiency $a >0$ and $\hat{a} > 0$: 
\begin{equation}
\label{eq: F and hatF}
F(a,\theta) = \sum_{j=0}^{N-1} \eta_je^{j\beta(a\theta+\delta)}, \qquad \hat{F}(\hat{a},\theta) = \sum_{j=0}^{N-1} \hat{\eta_j}e^{j\beta(\hat{a}\theta+\delta)}\quad \text{ and} \qquad G(a,\theta) = \sum_{j=0}^{N-1} e^{j\beta(a\theta+\delta)}.
\end{equation}

With these definitions, we can write the Social Welfare functions as follow:
\begin{subequations}
\begin{align}
SW(\theta) &= \frac{N^2}{2}\left(\delta+N\Delta - (1-a)\theta \right) \frac{F(a, \theta)}{G(a, \theta)},\label{sw1} \\
\widehat{SW}(\theta) &= \frac{N^2}{2} \big(\delta + N\Delta) \frac{F(\hat{a}, \theta)}{G(\hat{a}, \theta)} - \frac{N^2}{2}(1+\hat{a})\theta \frac{\hat{F}(\hat{a}, \theta)}{G(\hat{a}, \theta)}.\label{swb1}
\end{align}    
\end{subequations}

To prove Theorem \ref{thm:eff}, we will need two auxiliary lemmas. The first one expresses the difference between the two (rescaled) objective functions.
\begin{lemma}
For all $\theta \ge 0$, we have
\begin{equation}
\widehat{SW} \left( \frac{\theta}{\hat{a}} \right) - SW \left( \frac{\theta}{a} \right) = \frac{N^2\theta}{2G(1,\theta)} \left( \frac{1-a}{a}F(1,\theta) - \frac{1+\hat{a}}{\hat{a}}\hat{F}(1,\theta) \right)
\label{eq:sw_diff}
\end{equation}
\end{lemma}
\begin{proof}
From \eqref{sw1}-\eqref{swb1}, we have:
\[
SW \left( \frac{\theta}{a} \right) = \frac{N^2}{2} \left( \delta + N\Delta - \frac{1-a}{a}\theta \right) \frac{F(1, \theta)}{G(1, \theta)},
\]
and
\[
\widehat{SW} \left( \frac{\theta}{\hat{a}} \right) = \frac{N^2}{2} \big(\delta + N\Delta)\frac{F(1, \theta)}{G(1, \theta)} - \frac{N^2}{2}\frac{1+\hat{a}}{\hat{a}}\theta\frac{\hat{F}(1, \theta)}{G(1, \theta)}.
\]
From these expressions, by rearranging terms, then we obtain \eqref{eq:sw_diff} as claimed.
\end{proof}
The second lemma estimates the ratio between $F$ and $\hat{F} $ that appear in \eqref{eq:sw_diff}.
\begin{lemma}
    For all $a, \theta \in \mathbb{R}$, we have:
\[
\frac{\eta_0}{\eta_{N-1}} \le \frac{F(a, \theta)}{\hat{F}(a, \theta)} \le \frac{\eta_{N-1}}{\eta_0},
\]
or, equivalently:
\[
1 - \frac{N-2}{(N-1)(H_N + 1)}  \le  \frac{F(a, \theta)}{\hat{F}(a, \theta)} \le 1 + \frac{N-2}{(N-1)H_N + 1}
\] \label{lem:frac_ratio}
\end{lemma}
\begin{proof}
It follows from the formulas of $F$ and $\hat F$ in \eqref{eq: F and hatF} that
\[
\frac{F(a, \theta)}{\hat{F}(a, \theta)} = \frac{\sum_{j=0}^{N-1} \eta_j e^{j\beta(a\theta +\delta)}}{\sum_{j=0}^{N-1} \hat{\eta}_j e^{j\beta(a\theta+\delta)}}.
\]
Therefore, by the Generalized Mediant Inequality, the value of this ratio lies between the smallest and largest component fraction of the mediant
\[
\min_{j} \frac{\eta_j e^{j\beta(a\theta +\delta)}}{\hat{\eta}_j e^{j\beta(a\theta+\delta)}} \le \frac{F(a, \theta)}{\hat{F}(a, \theta)} \le \max_j \frac{\eta_j e^{j\beta(a\theta +\delta)}}{\hat{\eta}_j e^{j\beta(a\theta+\delta)}},
\]
that is
\[
\min_{j}\frac{\eta_j}{\hat{\eta}_j}\leq \frac{F(a, \theta)}{\hat{F}(a, \theta)} \le \max_{j}\frac{\eta_j}{\hat{\eta}_j}.
\]
A direct pairwise comparison shows that:
\[
\frac{\eta_0}{\eta_{N-1}} \le \frac{\eta_j}{\hat{\eta}_j} \le \frac{\eta_{N-1}}{\eta_0}
\]

for all $0 \le j \le N-1$. Therefore:
\[
\frac{\eta_0}{\eta_{N-1}} \le \frac{F(a, \theta)}{\hat{F}(a, \theta)} \le \frac{\eta_{N-1}}{\eta_0},
\]
which, after an algebraic manipulation step, yields the second desired inequality. Furthermore, these bounds are tight, as the ratio asymptotically approaches the lower and upper bounds as $x \to -\infty$ and $x \to +\infty$, respectively.
\end{proof}




We are now in the position to prove Theorem \ref{thm:eff}.
\begin{proof}[Proof of Theorem \ref{thm:eff}]

For $a \ge 1$, the right-hand side of equation (\ref{eq:sw_diff}) is strictly non-positive, implying that $\widehat{SW}(\theta) \le SW(\hat{a}\theta/a)$ is trivially true for all $\theta\ge 0$. Therefore, we restrict our focus to the case where $a < 1$. 

The condition for the (shifted) punishment to strictly outperform the (shifted) reward, $\widehat{SW}(\theta/\hat{a}) > SW(\theta/a)$, is equivalent to:
\begin{equation}
    \frac{1-a}{a}F(1,\theta) - \frac{1+\hat{a}}{\hat{a}}\hat{F}(1,\theta) > 0 \label{ineq:comprp}
\end{equation}
which, by letting $a^* = 1/a$ and $\hat{a}^* = 1/\hat{a}$, can be equivalently rewritten as a ratio of the polynomials
\begin{align}
    \frac{F(1, \theta)}{\hat{F}(1, \theta)} &> \frac{\hat{a}^* +1}{a^* - 1} \label{ineq:ratio_bound}.
\end{align}
By combining inequality \eqref{ineq:ratio_bound} with Lemma \ref{lem:frac_ratio}, we know that the ratio $\frac{F(a, \theta)}{\hat{F}(a, \theta)}$ is bounded above by $\frac{\eta_{N-1}}{\eta_0}$. Therefore, if we enforce:
\[
\frac{\hat{a}^*+1}{a^*-1} \ge \frac{\eta_{N-1}}{\eta_0}
\]
or equivalently:
\begin{equation}
    \eta_0 \hat{a}^* + \eta_0  \ge \eta_{N-1}a^* - \eta_{N-1}
    \label{ineq:eff_bound}
\end{equation}
then inequality \eqref{ineq:ratio_bound} can never be satisfied for any $\theta$. This guarantees $\widehat{SW}(\theta/\hat{a}) \le SW(\theta/a)$, or $\widehat{SW}(\theta) \le SW(\hat{a}\theta/a)$ globally. 

By some algebraic manipulation and substitution $a = 1/a^*$ and $\hat{a} = 1/\hat{a}^*$, inequality (\ref{ineq:eff_bound}) is equivalent to:
\[
a \ge \frac{\hat{a} \eta_{N-1}}{\eta_0 + \hat{a}(\eta_0 + \eta_{N-1})}
\]
which is exactly the claimed condition in Theorem \ref{thm:eff}.
\end{proof}
From Theorem \ref{thm:eff}, we have the following Corollary:

\begin{corollary} \label{col:thm3infb}
If it holds true that:
\[
a \ge \frac{\eta_{N-1}}{\eta_0 + \eta_{N-1}}
\]
then $\widehat{SW}(\theta) \le SW(\hat{a}\theta/a)$ for all $\theta \ge 0$, regardless of $\hat{a}$.
\end{corollary}
\begin{proof}
Notice that the function
\[
f(\hat{a}) = \frac{\hat{a} \eta_{N-1}}{\eta_0 + \hat{a}(\eta_0 + \eta_{N-1})} 
\]
is strictly increasing for all $\hat{a} > 0$. Therefore, taking the limit $\hat{a} \to +\infty$ yields the desired result.
\end{proof}



\subsection{Additional numerical simulations} \label{add-figures}
This section presents additional figures for the punishment case and the multi-objective comparison, to further illustrate the behaviour of the social welfare function. The supplementary results provide visual evidence for the key structural properties of the $SW$ curve across different game parameter settings. 

\begin{figure}[tbp]
    \centering
    \includegraphics[width=1\linewidth]{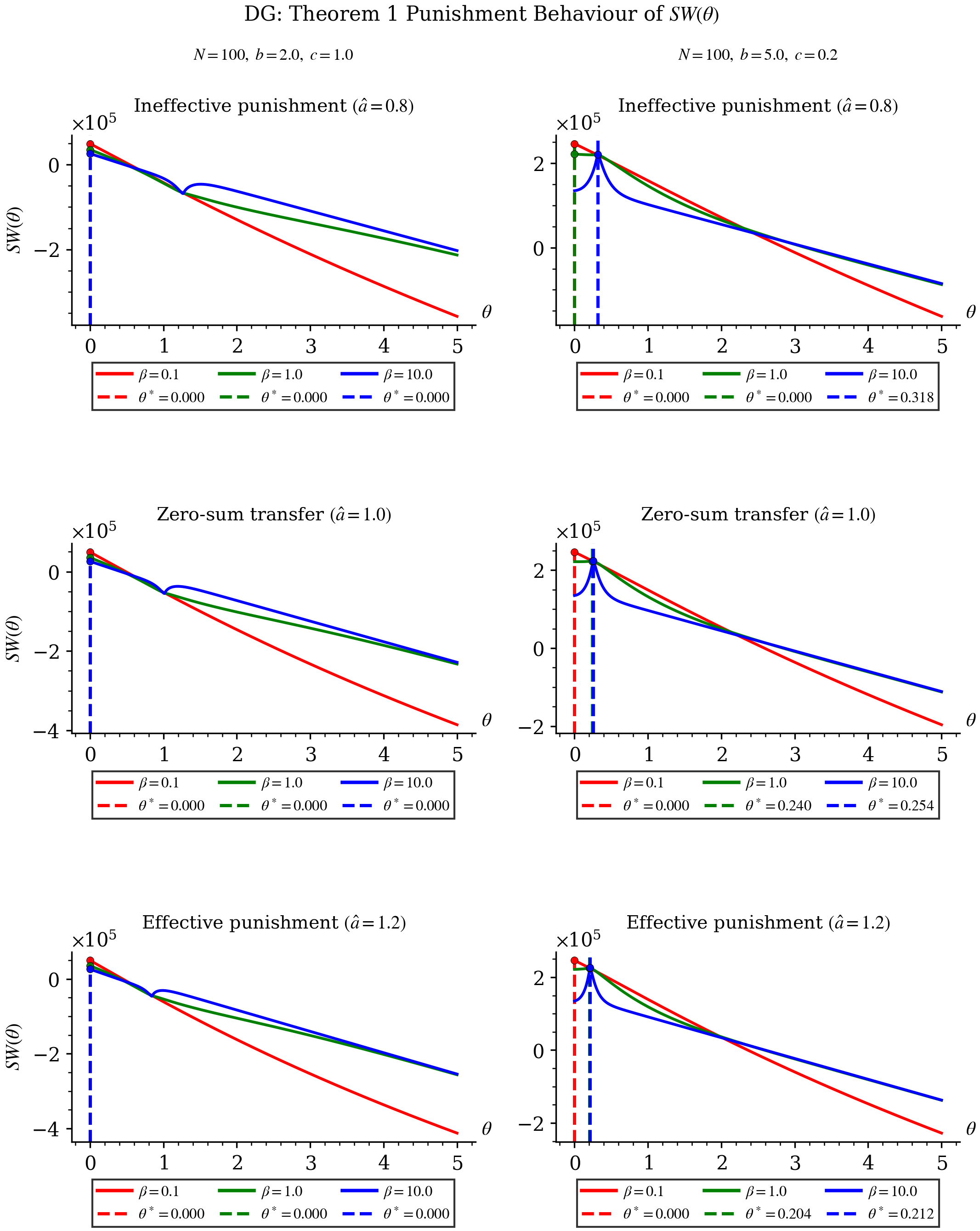}
    \caption{Numerical results for the Donation Game (DG) for the punishment-based policy. The $SW$ curves of varying selection intensity seemingly exhibit phase transitions and monotonic behaviours similar to the reward case, although with different thresholds for $\hat{a}$ and $\beta^*$.}
    \label{fig:thm1-punishment-dg}
\end{figure}

\begin{figure}[tbp]
    \centering
    \includegraphics[width=1\linewidth]{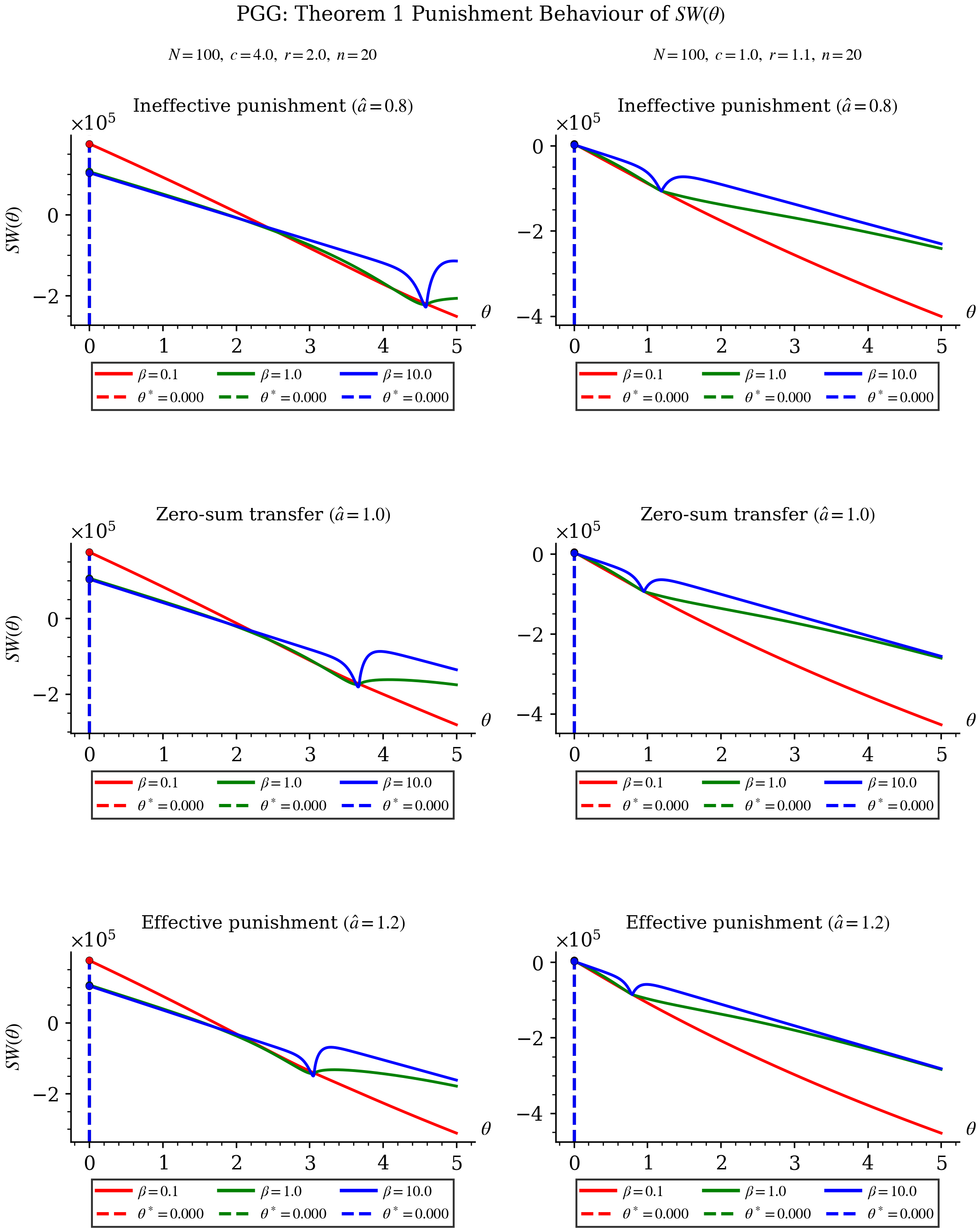}
    \caption{Numerical results for the Public Goods Game (PGG) for the punishment-based policy. The $SW$ curves of varying selection intensity seemingly exhibit phase transitions and monotonic behaviours similar to the reward case, although with different thresholds for $\hat{a}$ and $\beta^*$.}
    \label{fig:thm1-punishment-pgg}
\end{figure}

\begin{figure}[tbp]
    \centering
    \includegraphics[width=1\linewidth]{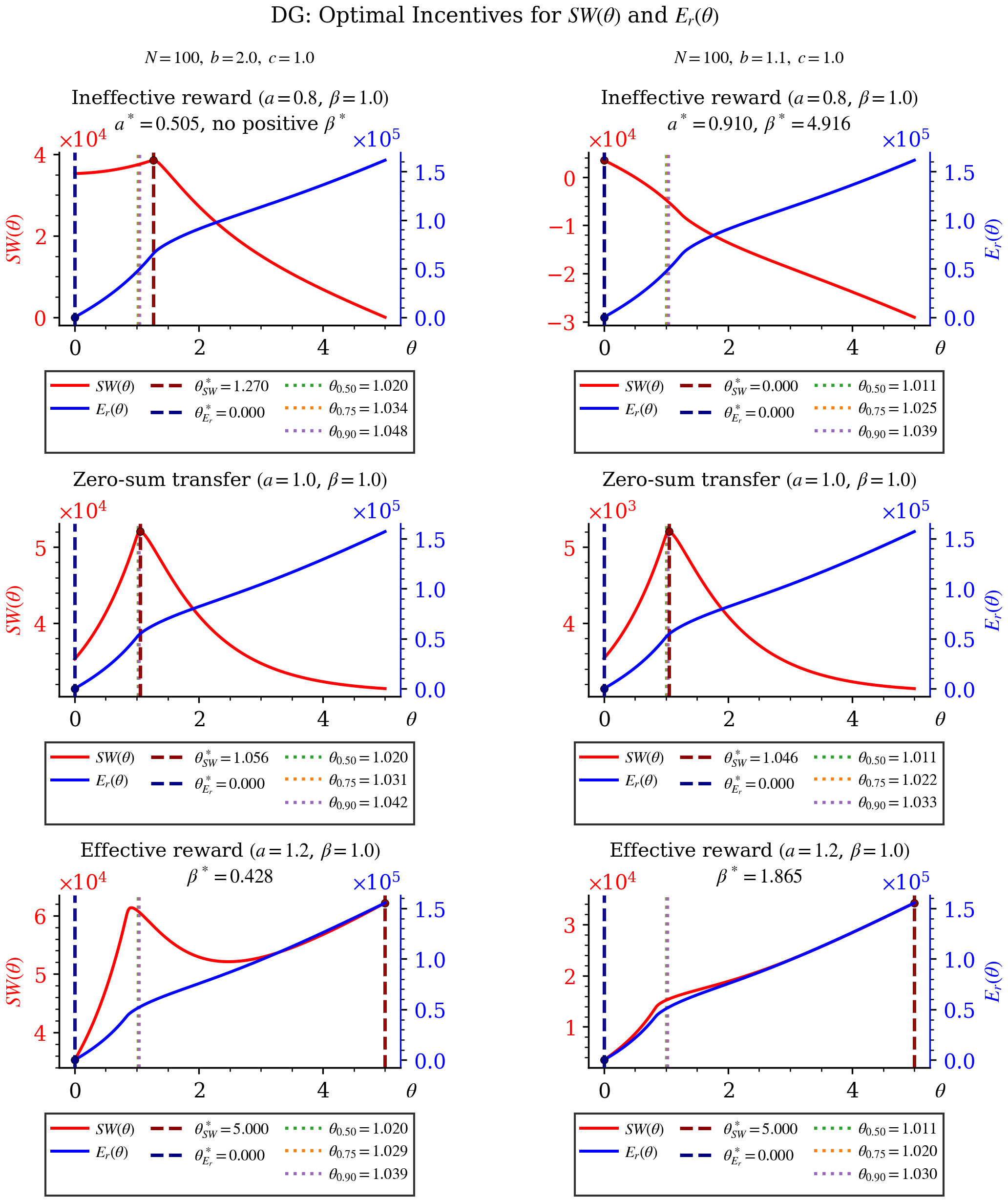}
    \caption{
    Multi-objective comparison of optimal incentive levels for the DG (reward case, $\beta=1.0$). The figures illustrate how the incentive values associated with maximising $SW(\theta)$, minimising $E_r(\theta)$, and satisfying frequency of cooperation thresholds vary across the three reward-efficiency regimes, namely ineffective reward $(a<1)$, zero-sum transfer $(a=1)$, and effective reward $(a>1)$. Vertical lines indicate the relevant optimal incentives and threshold values for target cooperation frequencies.
    }
\end{figure}

\begin{figure}[tbp]
    \centering
    \includegraphics[width=1\linewidth]{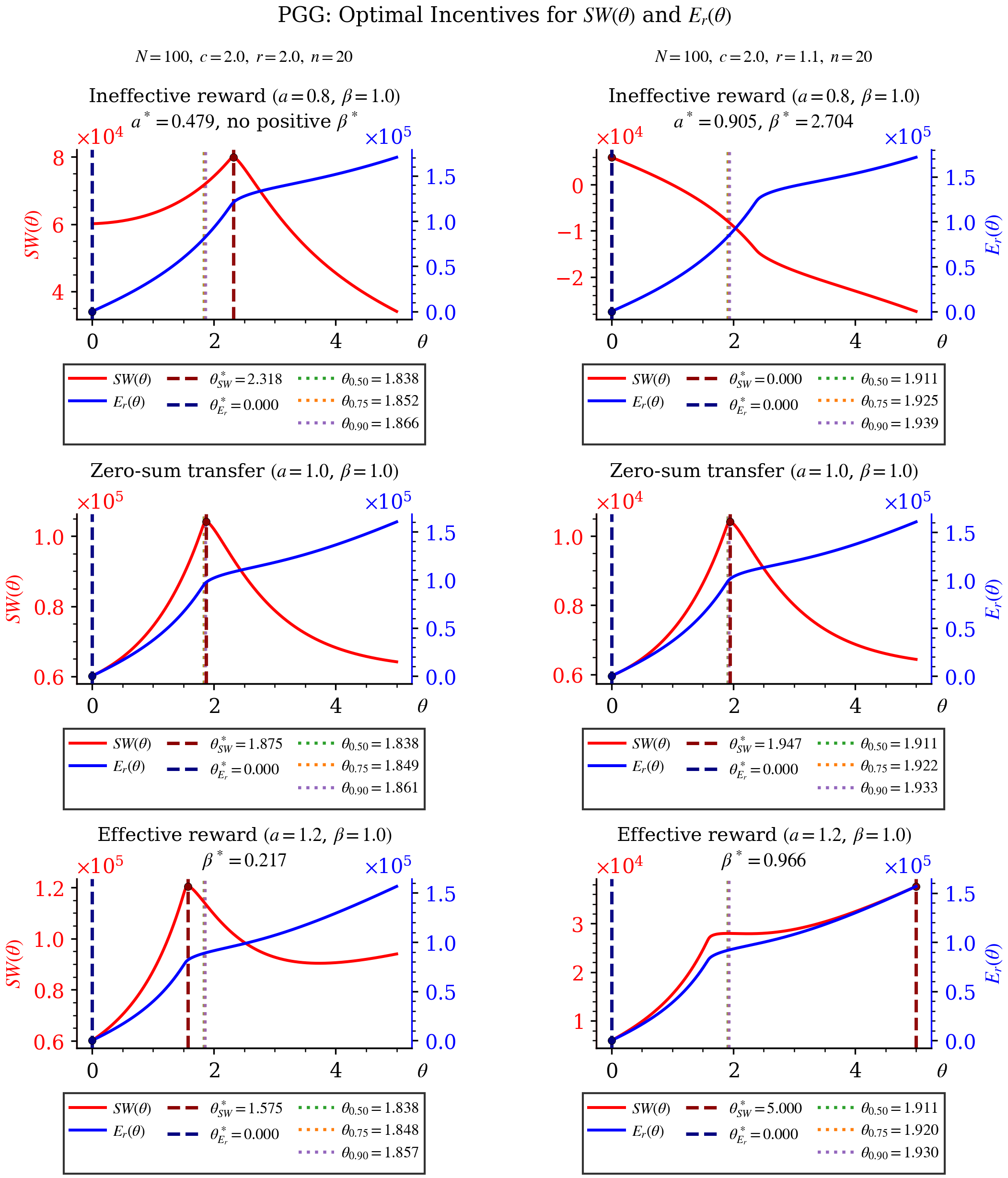}
    \caption{
    Multi-objective comparison of optimal incentive levels for the PGG (reward case, $\beta=1.0$). The figures illustrate how the incentive values associated with maximising $SW(\theta)$, minimising $E_r(\theta)$, and satisfying frequency of cooperation thresholds vary across the three reward-efficiency regimes, namely ineffective reward $(a<1)$, zero-sum transfer $(a=1)$, and effective reward $(a>1)$. Vertical lines indicate the relevant optimal incentives and threshold values for target cooperation frequencies.
    }
\end{figure}

\begin{figure}[tbp]
    \centering
    \includegraphics[width=1\linewidth]{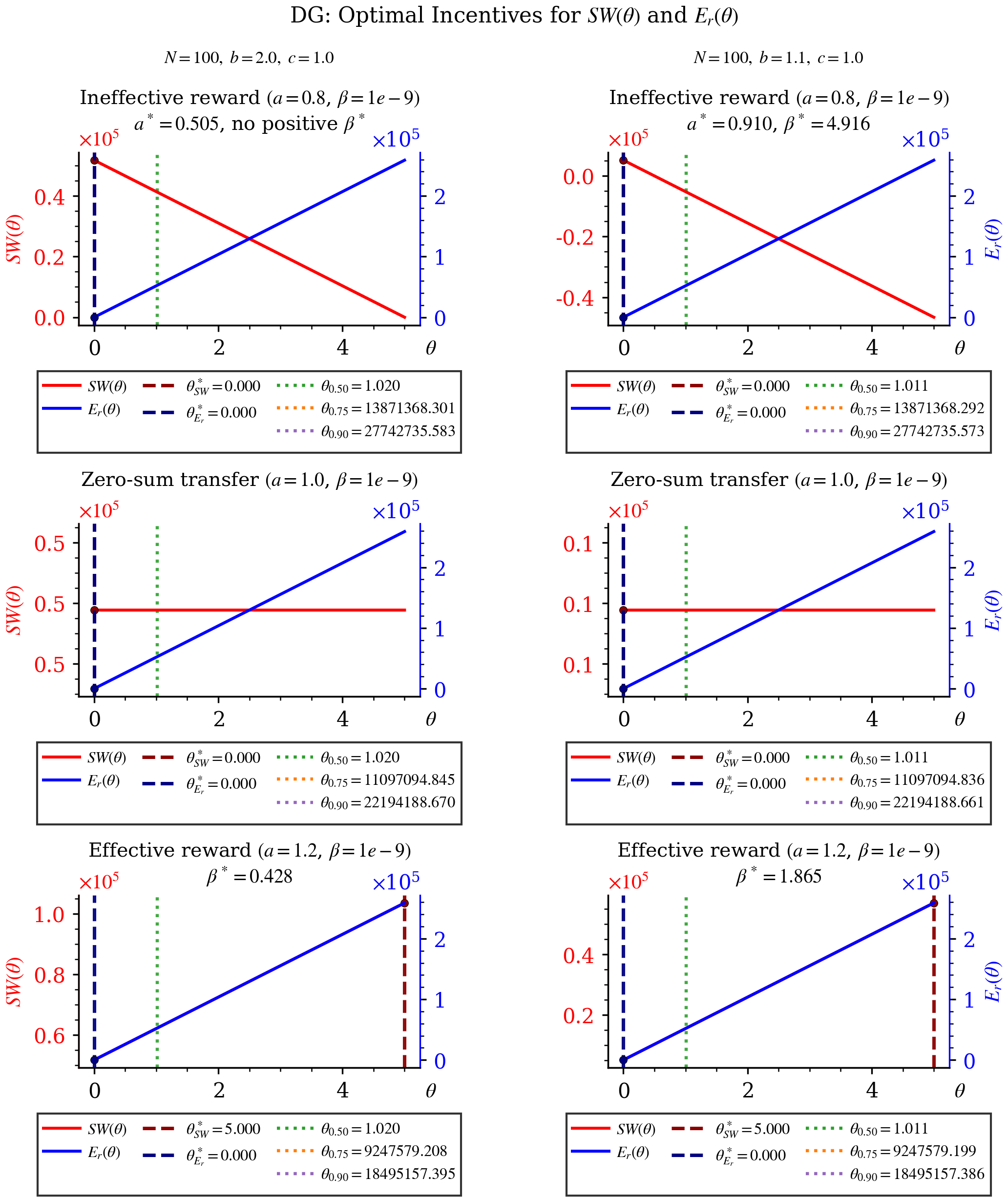}
    \caption{
    Multi-objective comparison of optimal incentive levels for the DG (reward case). Under the limit $\beta \to 0^+$, the minimum incentive levels required to achieve the target cooperation frequencies take extremely large values, causing the cooperation frequency thresholds to dominate the comparison with the social welfare and institutional cost objectives. As a result, the threshold requirements make the other objective-specific optima practically unattainable, except in the efficient-transfer regime $(a>1)$, where the social welfare optimum remains compatible with the required incentive level.
    }
\end{figure}

\begin{figure}[tbp]
    \centering
    \includegraphics[width=1\linewidth]{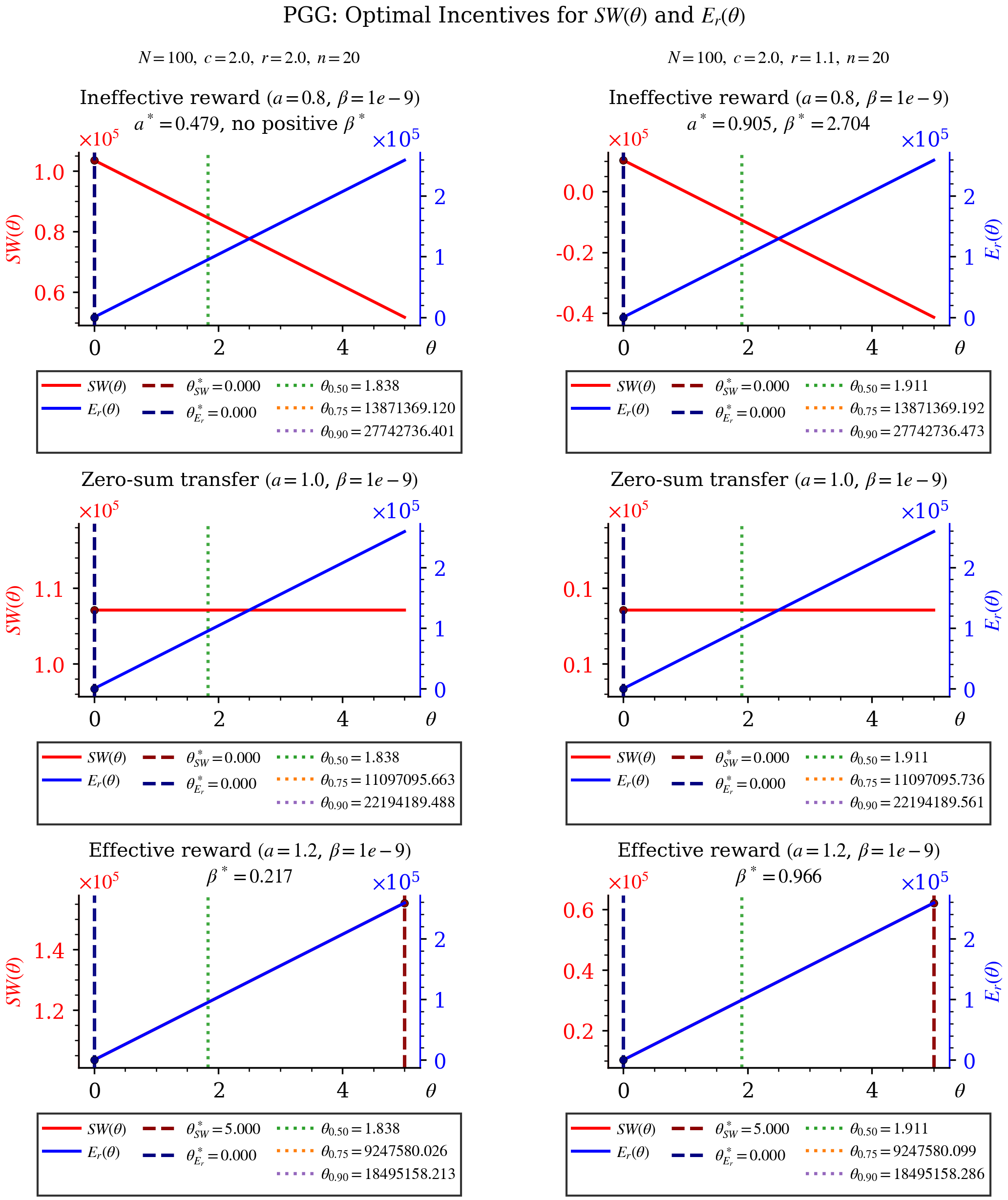}
    \caption{
    Multi-objective comparison of optimal incentive levels for the PGG (reward case). Under the limit $\beta \to 0^+$, the minimum incentive levels required to achieve the target cooperation frequencies take extremely large values, causing the cooperation frequency thresholds to dominate the comparison with the social welfare and institutional cost objectives. As a result, the threshold requirements make the other objective-specific optima practically unattainable, except in the efficient-transfer regime $(a>1)$, where the social welfare optimum remains compatible with the required incentive level.
    }
\end{figure}

\begin{figure}[tbp]
    \centering
    \includegraphics[width=1\linewidth]{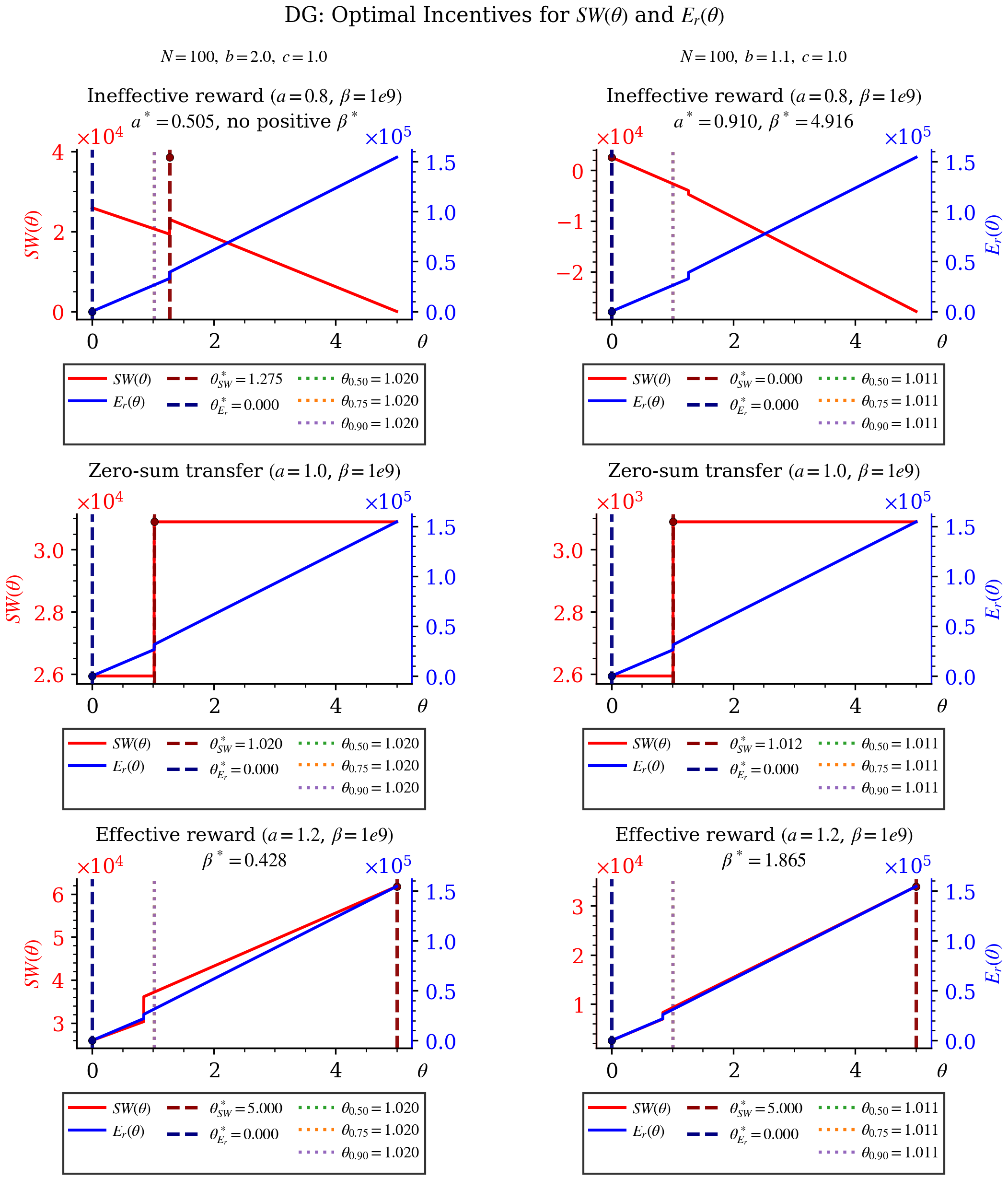}
    \caption{
    Multi-objective comparison of optimal incentive levels for the DG (reward case) under the strong-selection limit $\beta \to +\infty$. In this limit, the minimum incentive levels required to attain the target cooperation frequencies converge to $-\delta$. Consequently, the threshold requirements reduce the effective comparison to the two objective-specific optima: maximising $SW(\theta)$ and minimising $E_r(\theta)$ over the feasible interval $[-\delta,+\infty)$.\\\\For $a \le 1$, the two objectives are largely aligned, since the $SW(\theta)$ curve is  expected to be non-increasing over the feasible range. However, because an extremely large value of $\beta$ is still finite, the sharp near-vertical transition in $SW(\theta)$ can create a small discrepancy between the solutions for the social welfare and the institutional cost objectives. For $a>1$, the objectives become strongly conflicting: social welfare favours larger incentive levels, whereas institutional cost minimisation favours the smallest feasible incentive.
    }
\end{figure}

\begin{figure}[tbp]
    \centering
    \includegraphics[width=1\linewidth]{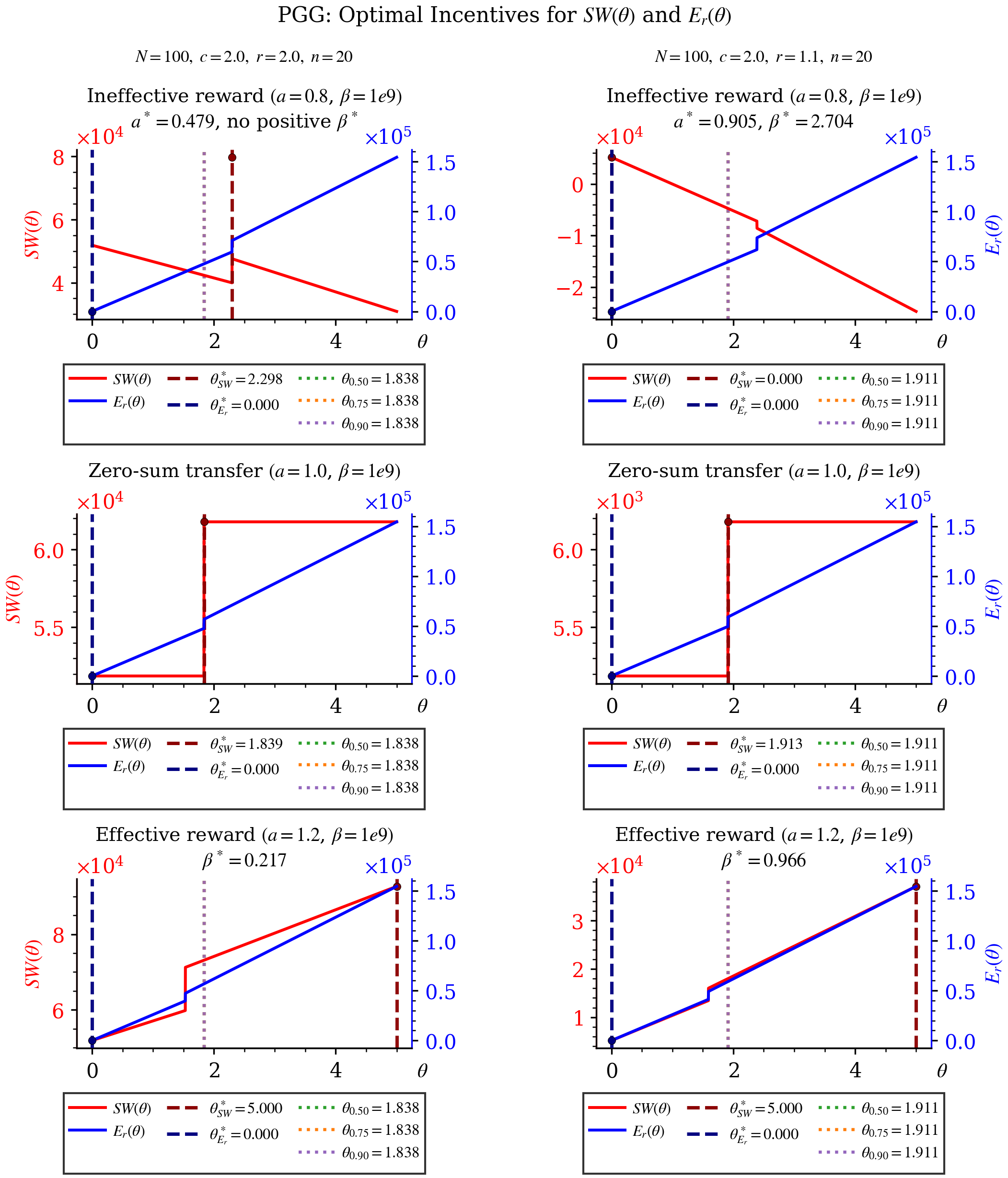}
    \caption{
    Multi-objective comparison of optimal incentive levels for the PGG (reward case) under the strong-selection limit $\beta \to +\infty$. In this limit, the minimum incentive levels required to attain the target cooperation frequencies converge to $-\delta$. Consequently, the threshold requirements reduce the effective comparison to the two objective-specific optima: maximising $SW(\theta)$ and minimising $E_r(\theta)$ over the feasible interval $[-\delta,+\infty)$.\\\\For $a \le 1$, the two objectives are largely aligned, since the $SW(\theta)$ curve is  expected to be non-increasing over the feasible range. However, because an extremely large value of $\beta$ is still finite, the sharp near-vertical transition in $SW(\theta)$ can create a small discrepancy between the solutions for the social welfare and the institutional cost objectives. For $a>1$, the objectives become strongly conflicting: social welfare favours larger incentive levels, whereas institutional cost minimisation favours the smallest feasible incentive.
    }
\end{figure}




\end{document}